\documentclass{IEEEtran}
\usepackage{graphics}
\usepackage{stfloats}
\usepackage{epsfig}
\usepackage{times}
\usepackage{amsthm}
\usepackage{amssymb}
\usepackage{amsmath}
\usepackage{indentfirst}
\usepackage{epstopdf}
\usepackage{multirow}
\usepackage{pdfsync}
\usepackage{float}
\usepackage{color}
\usepackage{enumerate}
\usepackage{subfigure}
\usepackage{comment}
\usepackage{cite}
\usepackage{caption}
\usepackage{algpseudocode}
\usepackage{algorithm}
\usepackage{algpseudocode}
\usepackage{flushend}
\DeclareMathOperator*{\minimize}{minimize}
\usepackage{tcolorbox}

\title{Adaptive Data-Driven Min-Max MPC for Linear Time-Varying Systems}

\author{Yifan~Xie, Julian~Berberich, Frank Allg\"{o}wer
	\thanks{F. Allg\"{o}wer is thankful that his work was funded by Deutsche Forschungsgemeinschaft (DFG, German Research
Foundation) under Germany’s Excellence Strategy - EXC 2075 - 390740016 and within grant AL 316/15-1-468094890.
The authors thank the International Max Planck Research School
for Intelligent Systems (IMPRS-IS) for supporting Yifan Xie.
	}
	\thanks{Yifan~Xie, Julian~Berberich, Frank Allg\"{o}wer are with the Institute for Systems Theory and Automatic Control, University of Stuttgart, 70550 Stuttgart, Germany. Julian~berberich is also with the Center for Integrated Quantum Science and Technology (IQST), University of Stuttgart, 70550 Stuttgart, Germany.
		{\tt\small  (email: \{yifan.xie, julian.berberich,
        frank.allgower\}@ist.uni-stuttgart.de}). }
}

\newtheorem{mythm}{Theorem}

\newtheorem{mylem}{Lemma}

\newtheorem{myexm}{Example}
\newtheorem{remark}{Remark}
\newtheorem{assum}{Assumption}

\begin{document}
	
	\maketitle

\begin{abstract}
In this paper, we propose an adaptive data-driven min-max model predictive control (MPC) scheme for discrete-time linear time-varying (LTV) systems.
We assume that prior knowledge of the system dynamics and bounds on the variations are known, and that the states are measured online.
Starting from an initial state-feedback gain derived from prior knowledge, the algorithm updates the state-feedback gain using online input-state data.
To this end, a semidefinite program (SDP) is solved to minimize an upper bound on the infinite-horizon optimal cost and to derive a corresponding state-feedback gain.
We prove that the resulting closed-loop system is exponentially stabilized and satisfies the constraints.
Further, we extend the proposed scheme to LTV systems with process noise.
The resulting closed-loop system is shown to be robustly stabilized to a robust positive invariant (RPI) set.
Finally, the proposed methods are demonstrated by numerical simulations.
\end{abstract}
\begin{IEEEkeywords}
Data-driven control, Linear time-varying (LTV) systems, Robust control, Linear matrix inequalities (LMIs), Adaptive control.
\end{IEEEkeywords}

\section{Introduction}\label{sec:1}
Data-driven control has received growing interest in the control community in recent years.
Direct data-driven control methods aim to design controllers directly from data, without first identifying a system model.
Several frameworks have been proposed to design data-driven controllers, i.e., data informativity \cite{van2023informativity, DBLSCT2025}, the behavioral approach \cite{markovsky2023datapower} and data-driven model predictive control (MPC) \cite{berberich2024overview}.
In the data informativity framework, a quadratic matrix inequality (QMI) can effectively characterize the set of systems that explain the data, and the focus is to design controllers for all systems in this set. 
This framework was initially developed for linear time-invariant (LTI) systems \cite{persis2020formulas,van2023quadratic}, and has been extended to LPV systems \cite{verhoek2025direct, mejari2023data,miller2022data,verhoek2024decoupling}, and polynomial systems \cite{martin2023inference, martin2023guarantees}.

Linear time-varying (LTV) systems are systems whose parameters change over time and capture the time-dependent dynamics of physical processes \cite{tsakalis1993linear}, i.e.,  system behavior evolves due to varying factors such as temperature, pressure, etc.
To address the challenges posed by time-varying parameters, several adaptive control schemes have been proposed \cite{chen2020adaptive,tanaskovic2019adaptive,bujarbaruah2020adaptive,kohler2021robust}.
In \cite{tanaskovic2019adaptive}, a robust adaptive MPC scheme is proposed for LTV systems, which employs a real-time set membership identification algorithm using input–output data and the available prior knowledge. 
Similarly, \cite{bujarbaruah2020adaptive} proposes a set-membership-based model adaptation algorithm to estimate and update the time-varying uncertainty along with an adaptive stochastic MPC framework for systems perturbed by an additive time-varying uncertainty and process noise.
A tube-based robust adaptive MPC scheme is proposed in \cite{kohler2021robust} for nonlinear systems subject to parametric uncertainty and additive disturbances.
These methods typically consider polytopic time-varying uncertainty and follow a two-step process: identifying a nominal model using data and prior information, and applying a tube-based adaptive MPC scheme, which is called as indirect data-driven control.

Several data-driven control methods have been proposed for LTV systems \cite{hu2024robust, iannelli2024hybrid, nortmann2023direct, liu2023online, kogan2025robust}.
In \cite{nortmann2023direct}, the authors assume access to multiple data sequences generated by an LTV system. 
The system is treated as a sequence of LTI models and design controllers independently for each model in the sequence.
Both prior knowledge and data are used to design an $H_\infty$ controller in \cite{kogan2025robust}.
Other works have focused on switched systems, which use online data to identify the active mode and apply mode-dependent stabilizing controllers \cite{rotulo2022online, eising2024data}.
Online data-driven control has also been explored for nonlinear systems \cite{bozza2024online, berberich2022linear} and LTI systems \cite{baros2022online, xie2026data} using real-time measurements for control adaptation. 
In \cite{qu2021stable}, an online control method is proposed for LTV systems under stochastic noise.
Among existing works, \cite{liu2023online} considers using the knowledge of the bounds on the variation of the system dynamics and online data to capture the current system dynamics, however, recursive feasibility is not guaranteed, and constraints are not considered.
Inspired by \cite{liu2023online}, we develop an adaptive data-driven min-max MPC framework for LTV systems, using prior knowledge of the system dynamics, knowledge on the variations of the system dynamics, and online data. 
Additionally, we aim to provide rigorous guarantees of recursive feasibility and constraint satisfaction.

Min-max MPC has been proposed to address systems with parametric uncertainties on the system dynamics or process disturbances \cite{kothare1996robust, bemporad2003minmax, lu2000quasi, limon2006input}.
In the min-max MPC framework, linear matrix inequality (LMI)-based methods have been proposed to derive a tractable state-feedback control law using prior knowledge on the parametric uncertainty set that robustly stabilizes the system.
The resulting control law minimizes the worst-case infinite-horizon cost w.r.t. parametric uncertainty \cite{kothare1996robust}.
Extensions of this framework, such as quasi min-max MPC have been proposed in \cite{lu2000quasi, yan2024improved}, where the system parameters are assumed to be known at the current time step but unknown in future steps.
Quasi-min-max MPC has also been proposed for LPV systems with polytopic parametric uncertainty and bounded rate of parameter changes \cite{li2010feedback, lu2000quasi2, park2004constrained}.
These formulations use known bounds on the rate of change of parameters to update the uncertainty polytope.
Recently, data-driven
min–max MPC schemes were proposed to use noisy input-state data to design controllers for various system classes, including LTI systems \cite{xie2026data, nguyen2025reducing}, LPV systems \cite{xie2024dataLPV}, and bilinear systems \cite{xie2025bilinear}.
The objective is to design a control law to minimize the worst-case cost over all system models consistent with the data, thereby ensuring robust performance and stabilization under uncertainty. 
However, data-driven min-max
MPC schemes for LTV systems remains an open challenge.

In this paper, we propose an adaptive data-driven min-max MPC scheme for LTV systems using prior knowledge of the system dynamics and online collected input-state data.
We assume that the variations of the system dynamics are bounded and the bounds are known.
Using the online input-state data and the known variation bounds on the system dynamics, we propose a data-driven characterization of the system matrices.
We then formulate an SDP program that minimizes an upper bound on the adaptive data-driven min-max MPC cost and derives a corresponding state-feedback controller. 
A receding-horizon algorithm is proposed to update the state-feedback control gain at each time steps.
The proposed scheme is recursively feasible, the resulting closed-loop system is exponentially stabilized to the origin, and satisfies the constraints.
Furthermore, we extend the adaptive data-driven min-max MPC scheme to LTV systems with bounded process noise.
Using the online noisy input-state data, a data-driven characterization is derived using the known bounds on the process noise and on the variation of the system dynamics.
We show that the resulting closed-loop system is robustly stabilized to a robust positive invariant (RPI) set.
Numerical examples show that the proposed schemes ensure stability and constraint satisfaction, and the online data improves the closed-loop performance.

The remainder of the paper is organized as follows.
Section~\ref{sec:2} introduces the problem setup.
In Section~\ref{sec:3}, we present a data-driven characterization of the system matrices based on online input-state data and assumptions on the rate of variation of the system matrices.
Section~\ref{sec:4} proposes an adaptive data-driven min-max MPC scheme for LTV systems. 
Closed-loop guarantees, including recursive feasibility, constraint satisfaction, and exponential stability are proved.
In Section~\ref{sec:5}, an adaptive data-driven min-max MPC scheme is proposed for LTV systems with process noise.
%The resulting closed-loop system is proven to be robustly stabilized to a robust positive invariant (RPI) set.
Numerical examples illustrating the performance of the proposed method are provided in Section~\ref{sec:6}.
Finally, Section~\ref{sec:7} concludes the paper and outlines directions for future research.

\textbf{Notations:} the set of natural numbers is denoted by $\mathbb{N}$.
The set of integers from $a$ to $b$ is denoted by $\mathbb{I}_{[a, b]}$.
For two natural numbers $m, n\in\mathbb{N}$, $m\mod{n}$ denotes the remainder of the division of $m$ by $n$.
The matrix $P\in\mathbb{R}^{n\times n}$ is called positive definite if $x^\top P x>0$ for all $x\in\mathbb{R}^n$ and positive semi-definite if $x^\top P x\geq 0$ for all $x\in\mathbb{R}^n$.
We write $P\succ 0$ if it is positive definite and $P\succeq 0$ if it is positive semi-definite.
We denote by $\lambda_{\min}(P)$ and $\lambda_{\max}(P)$ the minimum and maximum eigenvalues of the matrix $P$, respectively.
For a vector $x$ and a matrix $P\succeq 0$, we write $\|x\|_P^2=x^\top P x$.
%For two sets $S_1, S_2$, we denote the Minkowski sum as $S_1\oplus S_2=\{s_1+s_2:s_1\in S_1, s_2\in S_2\}$.
For matrices $A, B$ of compatible dimensions, we abbreviate $A B A^\top$ by $A B \begin{bmatrix}\star\end{bmatrix}^\top$.
%We denote by $0_{m\times n}$ the matrix with $m$-rows and $n$-columns, where every entry is zero.

\section{Problem Setup}\label{sec:2}
We consider an unknown discrete-time LTV system
\begin{equation}\label{system}
x_{t+1}=A_t x_t+B_t u_t,
%+\omega_t,
\end{equation}
where $x_t\in\mathbb{R}^n$ and $u_t\in\mathbb{R}^m$ denote the state and control input at time $t$, respectively.
The system matrices $A_t\in\mathbb{R}^{n\times n}$ and $B_t\in\mathbb{R}^{n\times m}$ are time-varying and unknown at each time step.
Without any prior knowledge of the system dynamics, designing a controller and establishing theoretical guarantees for the system \eqref{system} is difficult.
Therefore, we assume the prior knowledge of an uncertainty set that bounds the possible values of $A_t$ and $B_t$ at every time step, as specified in Assumption~\ref{assumption1}.
Moreover, we assume that the variations of the system matrices are bounded and the bound is known, as later stated in Assumption~\ref{assumption2}. 
The prior knowledge of an uncertainty set allows us to design a robust controller with theoretical guarantees for the closed-loop system.
The bounded variation assumption lets us adaptively redesign the controller over time using online input-state data to improve closed-loop performance.

\begin{assum}\label{assumption1}\upshape
The system matrices $(A_t, B_t)$ satisfy $(A_t, B_t)\in\Sigma_p$  for all $t\in\mathbb{N}$, where
\begin{equation}\label{assumption1:equation}
\Sigma_p=\{(A, B):\begin{bmatrix}
I &A &B
\end{bmatrix}
M_{p}
\begin{bmatrix}
I &A &B
\end{bmatrix}^\top\succeq 0 \}
\end{equation}
and $M_p=\begin{bmatrix}M_{p, 11} &M_{p, 12}\\M_{p, 12}^\top &M_{p, 22}\end{bmatrix}$ is a known matrix with $M_{p, 22}\prec 0$ and $M_{p, 11}-M_{p, 12}M_{p, 22}^{-1}M_{p, 12}^\top\succ 0$.
\end{assum}

\begin{remark}\upshape
Assumption~\ref{assumption1} is an ellipsoidal uncertainty bound which can represent model-based parametric uncertainty derived from first-principles modeling.
For example, if $A_t$ and $B_t$ are centered around a nominal estimate such that $\|\begin{bmatrix}A_t-\bar{A} &B_t-\bar{B}\end{bmatrix}\|\leq \alpha$ with $\alpha\geq 0$, $\bar{A}\in\mathbb{R}^{n\times n}, \bar{B}\in\mathbb{R}^{n\times m}$ for all $t\in\mathbb{N}$, then we can simply choose $M_{p, 11}=\alpha^2 I-\bar{A}\bar{A}^\top -\bar{B}\bar{B}^\top$, $M_{p, 12}=\begin{bmatrix}\bar{A} &\bar{B}\end{bmatrix}$ and $M_{p, 22}=-I$.
In some scenarios, prior knowledge of the time-varying matrices $A_t$ and $B_t$ is given in polytopic form \cite{kothare1996robust}. 
This polytopic uncertainty can be represented under Assumption~\ref{assumption1} by constructing an outer ellipsoidal approximation of the polytopic set.
\end{remark}

We define the time-varying perturbation of the system matrices between time $t_1$ and time $t_2$ as
\[\Delta A_{t_2}^{t_1}=A_{t_1}-A_{t_2}, \Delta B_{t_2}^{t_1}=B_{t_1}-B_{t_2},\]
where $t_1, t_2\in\mathbb{N}$.
In the following assumption, we assume that the variations of the system matrices are bounded and the bounds are known.

\begin{assum}\label{assumption2}\upshape
The variation of the system matrices satisfy $(\Delta A_{t}^{t-i},\Delta B_{t}^{t-i})\in\Pi_{i}$ for all $t\in\mathbb{N}, i\in\mathbb{I}_{[1, t]}$, where
\begin{equation}\label{assumption2:equation}\nonumber
    \Pi_{i}\!=\!\left\{\!(\Delta A, \Delta B):\!\!\begin{bmatrix}
        I 
        \!\!&\Delta A
        \!\!&\Delta B
    \end{bmatrix}M_{i}\begin{bmatrix}
        I 
        \!\!&\Delta A
        \!\!&\Delta B
    \end{bmatrix}^\top \!\!\succeq\! 0\right\}
\end{equation}
and $M_{i}\!\!=\!\!\begin{bmatrix}M_{i, 11} &M_{i, 12}\\M_{i, 12}^\top &M_{i, 22}\end{bmatrix}$ are known matrices with $M_{i, 22}\prec 0$ and $M_{i, 11}-M_{i, 12}M_{i, 22}^{-1}M_{i, 12}^\top\succ 0$ for all $i\in\mathbb{N}$.
%Moreover, $\Pi_i\in\Pi_{i+1}, \forall i\in\mathbb{I}_{[1,t]}$.
\end{assum}

Assumption \ref{assumption2} can model different forms of knowledge on the time-varying dynamics, e.g.,
\begin{enumerate}[(i)]
\item Lipschitz continuous dynamics: $M_{i, 11}=L_{i}^2 I$, $M_{i, 12}=0$ and $M_{i, 22}=-I$ with $L_{i}\in\mathbb{R}$ implies that the time-varying dynamics is Lipschitz continuous, i.e., $\|\begin{bmatrix}\Delta A^{t-i}_{t} &\Delta B^{t-i}_{t}\end{bmatrix}\|\leq L_{i}$ for all $t\in\mathbb{N}, i\in\mathbb{I}_{[1, t]}$.
In particular, if the variation of the system matrices satisfies $\|\begin{bmatrix}\Delta A_{t}^{t-i} &\Delta B_{t}^{t-i}\end{bmatrix}\|\leq \beta i$ for all $i\in\mathbb{I}_{[1, t]}$, then we can simply choose $L_{i}=\beta i$.
If the variation of the system matrices is centered around a nominal variation, i.e., $\|\begin{bmatrix}\Delta A^{t-i}_{t}-\Delta \bar{A} &\Delta B^{t-i}_{t}-\Delta \bar{B}\end{bmatrix}\|\leq L_{i}$ with $\Delta \bar{A}\in\mathbb{R}^{n\times n}, \Delta \bar{B}\in\mathbb{R}^{n\times m}$, then we can choose $M_{i, 11}=L_{i}^2 I-\Delta \bar{A}(\Delta \bar{A})^\top -\Delta \bar{B}(\Delta \bar{B})^\top$, $M_{i, 12}=\begin{bmatrix}\Delta \bar{A} &\Delta \bar{B}\end{bmatrix}$ and $M_{i, 22}=-I$. 
\item Periodic dynamics: $M_{i, 11}=0$, $M_{i, 12}=0$ and $M_{i, 22}=-I$ implies that the time-varying dynamics is periodic over a time interval $t_p=i\in\mathbb{N}$, i.e., $\Delta A^{t-t_p}_t=0$ and $\Delta B_{t}^{t-t_p}=0$ for all $t\in\mathbb{N}$.
\end{enumerate}

\begin{remark}\upshape
Assumption~\ref{assumption2} characterizes an ellipsoidal uncertainty bound on the variation of the system matrices. 
The conditions $M_{i, 11}-M_{i, 12}M_{i, 22}^{-1}M_{i, 12}^\top\succ 0$ and $M_{i, 22}\prec 0$ hold for various relevant example, e.g., Lipschitz continuous dynamics as above.
When the system varies periodically over a time interval $t_p$, one can choose $M_{t_p, 11}=\epsilon$ with a small positive value $\epsilon$ instead of $M_{t_p, 11}=0$ to satisfy this condition.
Although this introduces conservatism into the uncertainty description, it is necessary for numerical issues and data-driven characterization of the system matrices in the following section.
\end{remark}

\begin{remark}\upshape
Existing work considers a similar setting as Assumption~\ref{assumption2} on the variations of the system dynamics.
In \cite{liu2023online}, the variations of the system dynamics are bounded; however, only the Lipschitz continuous dynamics with $L_{i}=\beta i$ are considered.
Moreover, \cite{liu2023online} does not assume to have a prior knowledge on the system dynamics and, thus, no guarantees on recursive feasibility are given.
Periodic time-varying systems are considered in \cite{nortmann2023direct}, and \cite{hu2024robust} considers that the system matrices are in a convex hull of a set of vertices. 
Both \cite{nortmann2023direct} and \cite{hu2024robust} assume multiple offline trajectories of the system are available for each periodic time or vertices.
Those two work do not impose any bounds on the variation of the system dynamics and employ online data to improve the knowledge of the system dynamics.
In contrast, Assumption~\ref{assumption1} and \ref{assumption2} in our paper provide a more general framework that extends and unifies the problem setting considered in prior work on LTV systems.
Furthermore, unlike \cite{nortmann2023direct, hu2024robust}, our approach does not require offline data for controller design, but instead operates purely based on online measurements and prior knowledge.
\end{remark}

In the following example, we show how the Assumption~\ref{assumption1} and \ref{assumption2} can be applied to an academic example.
\begin{myexm}\upshape\label{example1}
Let us consider a discrete-time LTV system
in the form of \eqref{system} with 
\begin{equation}\nonumber
\begin{aligned}
A_t&\!=\!\begin{bmatrix}
1.1\!+\!0.2\sin{(\frac{\pi}{6}t)}\!\!\!\!\!\!&0.1 \!\!\!\!\!\!&0\\
0 \!\!\!\!\!\!&0.7\!+\!0.15\sin{(\frac{\pi}{6}t)}\!\!\!\!\!\!&-0.1\\
0 \!\!\!\!\!\!&0 \!\!\!\!\!\!&0.5\!+\!0.22\cos(\frac{\pi}{6}t)\end{bmatrix}, \\
B_t&\!=\!\begin{bmatrix}0.6 &0.1 &0.1\end{bmatrix}^\top
\end{aligned}
\end{equation}
for all $t\in\mathbb{N}$.
Suppose we have prior knowledge of the system dynamics, i.e., 
\[\|\begin{bmatrix}A_t-\bar{A} &B_t-\bar{B}\end{bmatrix}\|\leq \alpha.\]
with $\bar{A}=\begin{bmatrix}1.1 &0.1 &0\\0 &0.7 &-0.1\\ 0 &0 &0.5\end{bmatrix}$, $\bar{B}=\begin{bmatrix}0.6\\0.1\\0.1\end{bmatrix}$ and $\alpha=0.22$.
Then, Assumption~\ref{assumption1} holds with \begin{equation}\nonumber
\begin{aligned}
    M_{p, 11}&\!=\!\alpha^2 I\!-\!\bar{A}\bar{A}^\top \!-\!\bar{B}\bar{B}^\top\!\!=\!\!\begin{bmatrix}
    -1.5316 \!\!\!&-0.13 \!\!\!&-0.06\\
    -0.13 \!\!\!&-0.4616 \!\!\!&0.04\\ 
    -0.06 \!\!\!&0.04 \!\!\!&-0.2116\end{bmatrix},\\
    M_{p, 12}&\!=\!\begin{bmatrix}\bar{A} &\bar{B}\end{bmatrix}\!=\!\begin{bmatrix}1.1 &0.1 &0 &0.6\\0 &0.7 &-0.1 &0.1\\ 0 &0 &0.5 &0.1\end{bmatrix}
\end{aligned}
\end{equation} 
and $M_{p, 22}=-I$.
Furthermore, suppose we know a bound of the variation of the system dynamics, i.e., $\|\begin{bmatrix}\Delta A_{t}^{t-i} &\Delta B_{t}^{t-i}\end{bmatrix}\|\leq \beta i$ with $\beta=\frac{11}{300}\pi$ and we know that the system is periodic over the time interval $t_p=12$.
Then, Assumption~\ref{assumption2} holds with each of the individual matrices
\begin{equation}
\begin{aligned}\nonumber
    M_{i}&\!=\!\!\begin{bmatrix}
        \frac{121}{90000}\pi^2 (i\mod{12})^2I \!\!&0\\
        0 \!\!&-I
    \end{bmatrix}, \forall i\neq 12k, k\in\mathbb{N},\\
    M_{i}&\!=\!\!\begin{bmatrix}
        10^{-8}I \!\!&0\\
        0 \!\!&-I
    \end{bmatrix}, \forall i= 12k, k\in\mathbb{N}.
\end{aligned}
\end{equation}
\end{myexm}

So far, we have assumed that some prior model knowledge is available, as in Assumption~\ref{assumption1} and \ref{assumption2}. 
In the next section, we propose a data-driven characterization of the system dynamics using online collected data to improve our knowledge of the system dynamics. 
We aim to find a time-varying state-feedback control law that stabilizes the origin of the closed-loop system using both prior knowledge and online input-state data.
To evaluate the performance of the controller, we define the quadratic stage cost function
\[\ell(u, x)=\|u\|_R^2+\|x\|_Q^2,\]
where $R, Q\succ 0$.
Furthermore, we consider the following constraints on the closed-loop system
\begin{equation}\label{constraint}
\|C_x x_t+C_u u_t\|\leq 1, \forall t\in\mathbb{N},
\end{equation}
where $C_x\in\mathbb{R}^{n_c\times n}$ and $C_u\in\mathbb{R}^{n_c\times m}$.
This formulation generalizes individual ellipsoidal constraints on the input or state by choosing either $C_x$ or $C_u$ as 0.

\section{Data-Driven Characterization using Online Data}\label{sec:3}

In this section, we propose a method to characterize the set of consistent system matrices using the online input-state data using Assumption \ref{assumption2}.

At time $t$, we assume to have a sequence of input-state data from time $0$ to time $t$ as follows:
\begin{subequations}\label{data}
    \begin{align}
        X_{t}&=\begin{bmatrix}
            x_0 &x_1 &\ldots &x_{t-1} &x_t
        \end{bmatrix},\label{state}\\
        U_{t}&=\begin{bmatrix}
            u_0 &u_1 &\ldots &u_{t-1}
        \end{bmatrix}.\label{input}
    \end{align}
\end{subequations}

%The system dynamics at time $t-i$ can be written as
%\begin{equation}\nonumber
    %x_{t-i+1}=A_{t}x_{t-i}+B_{t}u_{t-i}+\Delta A^{t-i}_t x_{t-i}+\Delta B^{t-i}_t u_{t-i}.
%\end{equation}
The set of $(A, B)$ at time $t$ consistent with the data $x_{t-i}, u_{t-i}, x_{t-i+1}$ and Assumption \ref{assumption2} is defined as
\[\Sigma_{t, i}\!=\!\left\{(A, B): 
\begin{gathered}
    \exists (\Delta A^{t-i}_t, \Delta B^{t-i}_t)\in\Pi_i\text{ s.t. }\\x_{t-i+1}=Ax_{t-i}+Bu_{t-i}\\+\Delta A^{t-i}_t x_{t-i}+\Delta B^{t-i}_t u_{t-i}
\text{ holds}
\end{gathered}\right\}.\]
This set includes all possible system matrices at time $t$ for which there exists a variation of the system matrices $(\Delta A^{t-i}_t, \Delta B^{t-i}_t)$ satisfying Assumption~\ref{assumption2} such that the system dynamics \eqref{system} hold with the data $(x_{t-i}, u_{t-i}, x_{t-i+1})$.

Given the sequence of input-state data \eqref{data}, we obtain $t$ data-driven characterizations of the set of consistent system matrices $(A, B)$ at time $t$ via Assumption~\ref{assumption2}, i.e., $\Sigma_{t, i}, \forall i\in\mathbb{I}_{[1, t]}$.
We define the set of $(A, B)$ at time $t$ that are consistent with the sequence of input-state data \eqref{data} as
\[S_{t}=\{(A, B):(A, B)\in\Sigma_{t,i}, \forall i\in\mathbb{I}_{[1, t]}\},\]
which is the intersection of the sets of system matrices consistent with each data triple $(x_{t-i}, u_{t-i}, x_{t-i+1})$ for all $i\in\mathbb{I}_{[1, t]}$.

In the following lemma, we characterize the set $S_t$ using the data in \eqref{data}.
%based on a technical lemma \cite[Lemma 1]{xie2024dataLPV}.

\begin{mylem}\label{lemma1}\upshape
Suppose Assumption~\ref{assumption2} holds,  and $\begin{bmatrix}x_{t-i}\\ u_{t-i}\end{bmatrix}\neq 0$ for all $i\in\mathbb{I}_{[1, t]}$.
Then, the set $S_t$ is equal to
\begin{equation}\label{S_t}
\left\{\!(A, B)\!:
\begin{gathered}
\begin{bmatrix}
I \!\!&A \!\!&B
\end{bmatrix}
\Pi_{t}(\tau)
\begin{bmatrix}
I \!\!&A \!\!&B
\end{bmatrix}^\top\!\succeq\! 0,  \\
\forall \tau\!=\!(\tau_1, \ldots, \tau_t), \tau_i\!\geq\! 0, i\in\mathbb{I}_{[1, t]}
\end{gathered}\!
\right\}\!,
\end{equation}
where 
\begin{equation}\label{Pi_t}\nonumber
\Pi_{t}(\tau)\!=\!\sum_{i=1}^{t}\tau_i\!
\begin{bmatrix}
I &x_{t-i+1}\\
0 &-x_{t-i}\\
0 &-u_{t-i}
\end{bmatrix}
N_i
\begin{bmatrix}
I &x_{t-i+1}\\
0 &-x_{t-i}\\
0 &-u_{t-i}
\end{bmatrix}^\top.
\end{equation}
and $N_i$ is defined as in \eqref{Ni}.
\begin{figure*}
\begin{equation}\label{Ni}
    N_i=\begin{bmatrix}
        I \!\!&0\\
        \begin{bmatrix}x_{t-i}\\ u_{t-i}\end{bmatrix}^\top\!\!\! M_{i, 22}^{-1}M_{i, 12}^\top \!\!&I
    \end{bmatrix}^\top\!\!\!\begin{bmatrix}
        M_{i, 11}\!-\!M_{i, 12}(M_{i, 22})^{-1}M_{i, 12}^\top &0\\
        0 &(\begin{bmatrix}x_{t-i}\\ u_{t-i}\end{bmatrix}^\top \!\!M_{i, 22}^{-1}\begin{bmatrix}x_{t-i}\\ u_{t-i}\end{bmatrix})^{-1}
    \end{bmatrix}
    \begin{bmatrix}
        I \!\!&0\\
        \begin{bmatrix}x_{t-i}\\ u_{t-i}\end{bmatrix}^\top M_{i, 22}^{-1}M_{i, 12}^\top \!\!&I
    \end{bmatrix}
\end{equation}
\end{figure*}
\end{mylem}
\begin{proof}
Since Assumption~\ref{assumption2} holds, we have $(\Delta A^{t-i}_t, \Delta B^{t-i}_t)\in \Pi_i$. The virtual disturbance is defined as $\omega^{t-i}_t=\begin{bmatrix}\Delta A^{t-i}_t &\Delta B^{t-i}_t\end{bmatrix}\begin{bmatrix}x_{t-i}\\ u_{t-i}\end{bmatrix}$. 
We define the set 
\begin{equation}\label{set1}
    \left\{{\omega^{t-i}_t}^\top:\begin{bmatrix}
        I\\(\Delta A^{t-i}_t)^\top\\(\Delta B^{t-i}_t)^\top
    \end{bmatrix}^\top M_i
    \begin{bmatrix}
        I\\(\Delta A^{t-i}_t)^\top\\(\Delta B^{t-i}_t)^\top
    \end{bmatrix}\succeq 0\right\},
\end{equation}
which includes all possible realizations of ${\omega^{t-i}_t}^\top$ given that Assumption \ref{assumption2} holds. 
Since $\begin{bmatrix}x_{t-i}\\ u_{t-i}\end{bmatrix}\neq 0$ for all $i\in\mathbb{I}_{[1, t]}$, by applying Lemma \cite[Lemma 1]{xie2024dataLPV} for $\tilde{\Delta}=\begin{bmatrix}\Delta A^{t-i}_t &\Delta B^{t-i}_t\end{bmatrix}^\top$, $E=\begin{bmatrix}x_{t-i}\\ u_{t-i}\end{bmatrix}^\top$, $\tilde{G}=M_i$, $\hat{\Delta}={\omega^{t-i}_t}^\top$ and $Y=N_i$, the set in \eqref{set1} is equal to 
\begin{equation}\label{set2}
    \left\{{\omega^{t-i}_t}^\top:\begin{bmatrix}
        I\\{\omega^{t-i}_t}^\top 
    \end{bmatrix}^\top N_i
    \begin{bmatrix}
        I\\{\omega^{t-i}_t}^\top 
    \end{bmatrix}\succeq 0\right\},
\end{equation}
where $N_i$ is defined as in \eqref{Ni}.
The set in \eqref{set2} includes all possible realizations of ${\omega^{t-i}_t}^\top$ with $(\Delta A^{t-i}_t, \Delta B^{t-i}_t)\in\Pi_i$.
Replacing $\omega^{t-i}_t$ by $x_{t-i+1}-Ax_{t-i}-Bu_{t-i}$ in \eqref{set2}, we characterize all possible system matrices $(A, B)$ at time $t$ that are consistent with the data $(x_{t-i}, u_{t-i}, x_{t-i+1})$ and the Assumption~\ref{assumption2}, i.e., 
\begin{equation}\label{set3}
\Sigma_{t, i}\!\!=\!\!\left\{\!\!(A, B)\!\!:\!\!\!
\begin{bmatrix}
I\\ A^\top \\B^\top
\end{bmatrix}^\top\!\!\!
\begin{bmatrix}
I \!\!\!&x_{t-i+1}\\
0 \!\!\!&-x_{t-i}\\
0 \!\!\!&-u_{t-i}
\end{bmatrix}
\!\!N_i\!\!
\begin{bmatrix}
I \!\!\!&x_{t-i+1}\\
0 \!\!\!&-x_{t-i}\\
0 \!\!\!&-u_{t-i}
\end{bmatrix}^\top\!\!\!
\begin{bmatrix}
I\\ A^\top \\B^\top
\end{bmatrix}
\!\!\!\succeq \!0\!
\right\}.
\end{equation}
Using the characterization of $\Sigma_{t, i}$ as in \eqref{set3}, we obtain the characterization of the set $S_t$ as in \eqref{S_t}, analogous to \cite{berberich2023combining, bisoffi2021trade}.
\end{proof}

\begin{remark}\upshape
Lemma~\ref{lemma1} derives a data-driven characterization of the set $S_t$ using the collected data \eqref{data} from time $0$ to time $t$.
The key idea behind the data-driven characterization of the set $\Sigma_{t, i}$ is to obtain all possible values of the virtual disturbance $\omega^{t-i}_t=\Delta A^{t-i}_t x_{t-i}+ \Delta B^{t-i}_tu_{t-i}$ such that $(\Delta A^{t-i}_t, \Delta B^{t-i}_t)\in \Pi_i$. 
We then combine the data-driven characterization of $\Sigma_{t, i}$ for all $i\in\mathbb{I}_{[1, t]}$ using non-negative multipliers $\tau_i$ to derive the data-driven characterization of the set $S_t$.
In some cases, it is meaningful to discard input-state data beyond a specific past horizon, as older data may become uninformative.
For example, in LTV systems with Lipschitz continuous bound on the variation of the system dynamics, where $\|\begin{bmatrix}\Delta A_t^{t-i} &\Delta B_t^{t-i}\end{bmatrix}\|\leq \beta i$ for all $t\in\mathbb{N}$,
the Lipschitz bound becomes excessively large for large $i$, which makes very old data less useful for data-driven characterization.
To address this, we restrict the data-driven characterization of the set $S_t$ using only the most recent input-state data. Specifically, we define the truncated data as
\begin{subequations}\label{data_T}\nonumber
    \begin{align}
        X_{t, T}&=\begin{bmatrix}
            x_{t-T} &x_1 &\ldots &x_{t-1} &x_t
        \end{bmatrix},\label{state_T}\\
        U_{t, T}&=\begin{bmatrix}
            u_{t-T} &u_1 &\ldots &u_{t-1}
        \end{bmatrix},\label{input_T}
    \end{align}
\end{subequations}
where $T$ denotes the chosen data length.
The set $S_t$ is characterized as in \eqref{S_t} using the truncated data with
\begin{equation}\nonumber
\Pi_{t}(\tau)\!=\!\sum_{i=T}^{t}\tau_i\!
\begin{bmatrix}
I &x_{t-i+1}\\
0 &-x_{t-i}\\
0 &-u_{t-i}
\end{bmatrix}
N_i
\begin{bmatrix}
I &x_{t-i+1}\\
0 &-x_{t-i}\\
0 &-u_{t-i}
\end{bmatrix}^\top.
\end{equation}
This formulation allows us to improve computational efficiency by reducing the size of the optimization variables.
However, it is important to note that discarding data is not always desirable. 
For example, for systems with periodic dynamics, older data can still provide useful information about the system behavior at the current time step.
In such cases, it is beneficial to collect longer input-state data.
\end{remark}

\section{Adaptive Data-Driven Min-Max MPC Scheme}\label{sec:4}
In this section, we propose an adaptive data-driven min-max MPC scheme for LTV systems. 
In Section~\ref{sec:4.1}, a data-driven state-feedback controller is designed using prior information on the system dynamics to minimize a worst-case infinite-horizon cost.
This controller can serve as a backup controller to ensure recursive feasibility of the adaptive scheme.
In Section~\ref{sec:4.2}, an adaptive data-driven min-max MPC scheme is proposed by using the online collected data.
We reformulate the adaptive data-driven min-max MPC scheme as an SDP.
A receding-horizon algorithm is proposed to recompute the state-feedback gain at each time step.
Section~\ref{sec:4.3} establishes closed-loop guarantees such as recursive feasibility of the SDP problem, exponential stability of the resulting closed-loop system  and satisfaction of the constraints.

\subsection{Controller Design Using Prior Information}\label{sec:4.1}
As stated in Assumption~\ref{assumption1}, the system matrices $(A_t, B_t)$ at all time steps are assumed to lie within the set $\Sigma_p$.
In this section, we design a controller based on this prior information, which will then be used as part of the adaptive controller in the subsequent sections.

At time $t=0$, given prior knowledge of the system dynamics \eqref{assumption1:equation} in Assumption \ref{assumption1} and an initial state $x_0\in\mathbb{R}^n$, the min-max MPC optimization problem using prior knowledge is formulated as follows:
\begin{subequations}\label{mpc_initial}
\begin{align}
J_\infty^\star(x_0)&=\min_{\bar{u}(0)}\max_{(A, B)\in\Sigma_{p}}\sum_{k=0}^{\infty}\ell(\bar{u}_k(0), \bar{x}_k(0))\label{mpc:initial_obj}\\
\text{s.t.}\quad &\bar{x}_{k+1}(0)=A\bar{x}_k(0)+B\bar{u}_k(0),\label{mpc:initial_con1}\\
&\bar{x}_0(0)=x_0, \label{mpc:initial_con2}\\
    &\|C_x \bar{x}_k(0)+C_u \bar{u}_k(0)\|\leq 1, \forall (A, B)\in\Sigma_p, k\in\mathbb{N}.\label{mpc:initial_con3}
\end{align}
\end{subequations}
In problem \eqref{mpc_initial}, $\bar{u}_k(0)$ and $\bar{x}_k(0)$ are the predicted input and state at time $k$ given the initial state $x_0$ at time $t=0$.
The optimization variable $\bar{u}(0)=\{\bar{u}_0(0), \ldots, \bar{u}_k(0)\}$ with $k\rightarrow \infty$ is a sequence of predicted inputs. 
The objective is to minimize the worst-case cost over all $(A, B)\in\Sigma_p$.
We initialize $\bar{x}_0(0)$ as the measured state $x_0$ in \eqref{mpc:initial_con2}.
The predicted input and state satisfy the constraint for any system dynamics \eqref{mpc:initial_con1} with $(A, B)\in\Sigma_p$.

The min-max MPC problem \eqref{mpc_initial} is intractable because of the min-max formulation of the infinite-horizon cost.
We restrict our attention to find a state-feedback control law in the form $\bar{u}_k(0)=F\bar{x}_k(0)$ for all $k\in\mathbb{N}$.
In the following, we formulate an SDP that derives an upper bound on the optimal cost of \eqref{mpc_initial} and determine a corresponding state-feedback gain.

At time $t=0$, given prior knowledge of the system dynamics \eqref{assumption1:equation} in Assumption \ref{assumption1} and an initial state $x_0\in\mathbb{R}^n$, the SDP is formulated as follows:
\begin{subequations}\label{sdp_initial}
\begin{align}
    &\minimize\limits_{\gamma>0, H\in\mathbb{R}^{n\times n}, L\in\mathbb{R}^{m\times n}, \tau_p\geq 0}\gamma\label{sdp_initial:obj}\\
    \text{s.t. }
    &\begin{bmatrix}1 &x_0^\top\\
x_0 &H\end{bmatrix}\succeq 0, \label{sdp_initial:con1}\\
    &\begin{bmatrix}
        \begin{bmatrix}
            -H \!\!&0\\
            0 \!\!&0
        \end{bmatrix}\!+\Pi_{p} &
        \begin{bmatrix}
            0\\
            H\\
            L
        \end{bmatrix}
        \!\!& 0\\
        \begin{bmatrix}
            0 &H &L^\top
        \end{bmatrix} &-H \!\!&\Phi^\top\\
        0 &\Phi \!\!& -\gamma I
    \end{bmatrix}\prec 0, \label{sdp_initial:con2}\\
    &\begin{bmatrix}
        H &(C_xH+C_uL)^\top\\
        (C_xH+C_uL) &I
    \end{bmatrix}\succeq 0, \label{sdp_initial:con3}
\end{align}
\end{subequations}
where $\Pi_{p}=\tau_p M_{p}, \Phi=\begin{bmatrix}M_R L\\ M_Q H\end{bmatrix}$, $M_R^\top M_R=R$ and $M_Q^\top M_Q=Q$.
We denote the optimal solution of \eqref{sdp_initial} by $\gamma_p^\star, H_p^\star, L_p^\star, \tau_p^\star$.
The corresponding optimal state-feedback gain is given by $F_p^\star=L_p^\star(H_p^\star)^{-1}$ and we define $P_p^\star=\gamma_p^\star (H_p^\star)^{-1}$.

In the following theorem, we show that the SDP problem \eqref{sdp_initial} provides an upper bound on the optimal cost of the problem \eqref{mpc_initial}.
\begin{mythm}
\label{theorem1}\upshape
Given a state $x_0\in\mathbb{R}^n$, suppose Assumption \ref{assumption1} holds, and 
there exist $\gamma>0, H\succ 0$,  $L\in\mathbb{R}^{m\times n}, \tau_p\geq 0$ such that the problem \eqref{sdp_initial} is feasible.
Let $P=\gamma H^{-1}$.
Then, 
\begin{enumerate}[(i)]
    \item the optimal cost of \eqref{mpc:initial_obj}-\eqref{mpc:initial_con2} is guaranteed to be at most $\|x_0\|_{P}^2$ and $\|x_0\|_{P}^2$ is upper bounded by $\gamma$, i.e.,
$J^\star_\infty (x_0)\leq \|x_0\|_{P}^2\leq\gamma$;
    \item the constraint \eqref{mpc:initial_con3} is satisfied by applying $\bar{u}_k(0)=F\bar{x}_k(0)$ with $F=LH^{-1}$ to the system \eqref{mpc:initial_con1} for all $k\in\mathbb{N}$.
\end{enumerate}
\end{mythm}
\begin{proof}
The proof of (i) is analogous to that of Theorem~1 in \cite{xie2024minmax}, so we omit the proof here.
%According to Theorem 1 in \cite{xie2024minmax}, we can prove that $\gamma$ is an upper bound of the optimal cost of \eqref{mpc:initial_obj}-\eqref{mpc:initial_con2}.
We now prove that the satisfaction of 
\eqref{sdp_initial:con3} ensures the satisfaction of \eqref{mpc:initial_con3}.

Using the Schur complement, \eqref{sdp_initial:con3} is equivalent to
\begin{equation}\label{input:nominal_lmi4}
        H-(C_x H+C_uL)^\top (C_x H+C_uL)\succeq 0.
\end{equation}
Using $H=\gamma P^{-1}$, $F=LH^{-1}$ and multiplying both sides of \eqref{input:nominal_lmi4} with $\gamma^{-1}P$, the inequality \eqref{input:nominal_lmi4} is equivalent to
\begin{equation}
    \gamma^{-1} P-(C_x+C_uF)^\top (C_x+C_uF)\succeq 0.\label{input:lmi5}
\end{equation}
The following inequality holds by using \eqref{input:lmi5} and $1-\gamma^{-1}\gamma=0$
\begin{equation}\label{input:nominal_lmi1}
    \begin{bmatrix}
        -(C_x+C_uF)^\top (C_x+C_uF)&0\\
        0&1
    \end{bmatrix}-\gamma^{-1}
    \begin{bmatrix}
        -P &0\\
        0 &\gamma
    \end{bmatrix}\succeq 0,
\end{equation}
Pre-multiplying \eqref{input:nominal_lmi1} with $\begin{bmatrix}x^\top &1\end{bmatrix}$ and post-multiplying \eqref{input:nominal_lmi1} with $\begin{bmatrix}x^\top &1\end{bmatrix}^\top$, we obtain
\begin{equation}\label{thm11}
    \begin{bmatrix}
        x\\
        1
    \end{bmatrix}^\top\!\!\!\!\left (\!
    \begin{bmatrix}
        -(C_x\!+\!C_uF)^\top \!(C_x\!+\!C_uF)\!\!\!&0\\
        0\!\!\!&1
    \end{bmatrix}\!\!-\!\gamma^{-1}\!\!\!
    \begin{bmatrix}
        -P \!\!&0\\
        0 \!\!&\gamma
    \end{bmatrix}\!\right )\!\!\!
    \begin{bmatrix}
        x\\
        1
    \end{bmatrix}\!\geq\! 0,
\end{equation}
%Since $P_p^\star$, $\gamma_p^\star$, $F_p^\star$ is the optimal solution of \eqref{sdp_initial}, the inequality \eqref{input:lmi6} holds for $P=P_p^\star$, $\gamma=\gamma_p^\star$, $F=F_p^\star$.
This implies that, for any state $x$ satisfying
\begin{equation}\label{thm12}\nonumber
    \begin{bmatrix}
        x\\
        1
    \end{bmatrix}^\top
    \begin{bmatrix}
        \gamma^{-1} P &0\\
        0 &1
    \end{bmatrix}
    \begin{bmatrix}
        x\\
        1
    \end{bmatrix}
    \geq 0,
\end{equation}
the following inequality holds
\begin{equation}\label{thm13}\nonumber
    \begin{bmatrix}
        x\\
        1
    \end{bmatrix}^\top
    \begin{bmatrix}
        -(C_x+C_uF)^\top (C_x+C_uF) &0\\
        0 &1
    \end{bmatrix}
    \begin{bmatrix}
        x\\
        1
    \end{bmatrix}
    \geq 0.
\end{equation}
Thus, we have 
\begin{equation}\nonumber
\max\limits_{x\in\mathcal{E}}\|(C_x+C_uF)x\|\leq 1,
\end{equation}
where $\mathcal{E}=\{x:\|x\|_{P}^2\leq \gamma\}$.
Similar to the proof of Lemma~1 in \cite{xie2024minmax}, the set $\mathcal{E}$ is a RPI set for the system $x_{t+1}=(A+BF_p^\star)x_t$ with any $(A, B)\in \Sigma_{p}$.
Since $x_0=\bar{x}_0(0)\in\mathcal{E}$, the constraint \eqref{mpc:initial_con3} is fulfilled for all $(A, B)\in \Sigma_{p}$, $k\in\mathbb{N}$.
\end{proof}

\begin{remark}\label{remark5}\upshape
Problem~\eqref{sdp_initial} provides an upper bound on the optimal cost of the problem \eqref{mpc_initial} and a corresponding state-feedback control gain.
Similar to the approach in \cite{xie2024minmax}, we can show that the resulting closed-loop system converges exponentially to the origin and satisfies the constraint by applying the state-feedback gain $F_p^\star$.
The resulting data-driven state-feedback controller can be interpreted as a robust controller for a system with parametric uncertainty.
\end{remark}

\begin{remark}\label{remark6}\upshape
The controller design method in this subsection follows that of \cite{xie2024minmax} with a key difference that a more general input and state constraint is considered.
In \cite{kothare1996robust}, a robust controller design based on LMIs for LPV systems using model knowledge is propose. However, while \cite{kothare1996robust} addresses polytopic parametric uncertainty, the controller design in this subsection considers ellipsoidal parametric uncertainty.
\end{remark}

\subsection{Adaptive Data-Driven Min-Max MPC}\label{sec:4.2}

The performance of the state-feedback control law designed using prior information may be suboptimal, as it is intended to stabilize the system over the parametric uncertainties within the set $\Sigma_p$. 
Since the system dynamics are time-varying, it is desirable to obtain more accurate information about the system dynamics at the current time step and design a time-varying controller accordingly.
In this subsection, we present an adaptive data-driven min-max MPC scheme using online collected data to improve closed-loop performance.

At time $t$, we collect a sequence of input-state data from time $0$ to time $t$ as in \eqref{data}.
Recall that the set of system matrices at time $t$ consistent with the data \eqref{data} is denoted by $S_{t}$, which is characterized using Lemma~\ref{lemma1}.
The set of system matrices at time $t$ consistent with both the prior knowledge and the online input-state data is denoted as $S_t\cap \Sigma_{p}$.
Given $S_t\cap \Sigma_p$ and the current state $x_t\in\mathbb{R}^n$, the adaptive data-driven min-max MPC optimization problem is formulated as follows:

\begin{subequations}\label{mpc_onestep}
\begin{align}
J^\star(x_t)\!=\!&\min_{\bar{u}_0(t)=F \bar{x}_0(t)}\!\max_{(A, B)\in S_t\cap\Sigma_{p}}\!\!\!\ell(\bar{u}_0(t), \bar{x}_0(t))\!+\!\|\bar{x}_1(t)\|_{P_p^\star}^2\label{mpc:onestep_obj}\\
\text{s.t.}\quad &\bar{x}_1(t)=A\bar{x}_0(t)+B\bar{u}_0(t),\label{mpc:onestep_con1}\\
&\bar{x}_0(t)=x_t,\label{mpc:onestep_con2}\\
&\|C_x \bar{x}_0(t)+C_u \bar{u}_0(t)\|\leq 1.\label{mpc:onestep_con3}
\end{align}
\end{subequations}
In problem \eqref{mpc_onestep}, $\bar{u}_0(t)$ and $\bar{x}_0(t)$ denote the predicted input and state at time $t$, and $\bar{x}_1(t)$ denotes the predicted state at time $t+1$ given the measured state $x_t$.
The objective is to find a state-feedback control law that minimizes the worst-case value of a one-step stage cost plus a terminal cost over all $(A, B)\in S_t\cap \Sigma_p$.
The terminal cost function is defined over the predicted state $\bar{x}_1(t)$ using the matrix $P_p^\star$ obtained from the optimal solution of problem \eqref{sdp_initial}.
We initialize $\bar{x}_0(t)$ as the measured state $x_t$ in \eqref{mpc:onestep_con2}.
In constraint \eqref{mpc:onestep_con3}, the predicted input and state at time $t$ satisfy the ellipsoidal constraint in \eqref{constraint}.

\begin{remark}\upshape
Problem \eqref{mpc_onestep} minimizes a worst-case cost of a one-step stage cost over the set $S_t\cap \Sigma_p$ plus a terminal cost function that serves as an upper bound for the worst-case cost of an infinite horizon from time $t+1$.
In contrast, problem \eqref{mpc_initial} optimizes an infinite-horizon cost over the prior uncertainty set $\Sigma_p$.
By redesigning the controller at each time step to account for the uncertainty set $S_t\cap \Sigma_p$, the closed-loop performance can be improved.
This is because the considered parametric uncertainty in problem \eqref{mpc_onestep} is more accurately characterized and is potentially smaller than that only using the prior information.
In \eqref{mpc:onestep_obj}, the one-step stage cost is used rather than a multi-step stage cost. 
This choice avoids the need for predicting future time-varying system matrices, which would introduce substantial conservatism in the robust formulation.
\end{remark}

Similar to Section~\ref{sec:4.1}, we formulate an SDP to compute an upper bound on the optimal cost of \eqref{mpc_onestep} and a corresponding state-feedback control law.
At time $t$, given $S_t\cap \Sigma_p$, the matrix $P_p^\star$ and the current state $x_t\in\mathbb{R}^n$, the SDP is formulated as follows:
\begin{subequations}\label{sdp_onestep}
\begin{align}
    &\minimize\limits_{\gamma>0, H\in\mathbb{R}^{n\times n}, L\in\mathbb{R}^{m\times n}, \Gamma\in\mathbb{R}^{t+1}}\gamma\label{sdp_onestep:obj}\\
    \text{s.t. }
    &\begin{bmatrix}1 &x_t^\top\\
x_t &H\end{bmatrix}\succeq 0, \label{sdp_onestep:con1}\\
    &\begin{bmatrix}
        \begin{bmatrix}
            -\gamma (P_p^\star)^{-1} \!\!\!&0\\
            0 \!\!\!&0
        \end{bmatrix}\!+\!\Pi_t(\tau)\!+\!\Pi_{p} &
        \begin{bmatrix}
            0\\
            H\\
            L
        \end{bmatrix}
        & 0\\
        \begin{bmatrix}
            0 &H &L^\top
        \end{bmatrix} &-H &\Phi^\top\\
        0 &\Phi & -\gamma I
    \end{bmatrix}\prec 0, \label{sdp_onestep:con2}\\
    &\Gamma\!=\!(\tau_p, \tau_1, \ldots, \tau_t), \tau_p\!\geq\! 0, \tau_i\!\geq\! 0, \forall i\!\in\!\mathbb{I}_{[1, t]},\label{sdp_onestep:con3}\\
    &\begin{bmatrix}
        H &(C_xH+C_uL)^\top\\
        (C_xH+C_uL) &I
    \end{bmatrix}\succeq 0,\label{sdp_onestep:con4}
\end{align}
\end{subequations}
where $\Phi=\begin{bmatrix}M_R L\\ M_Q H\end{bmatrix}$, $M_R^\top M_R=R$, $M_Q^\top M_Q=Q$, $\Pi_{p}=\tau_p M_{p}$ and $\Pi_t(\tau)$ is defined as in Lemma~\ref{lemma1}.
The optimal solution of problem \eqref{sdp_onestep} depends on the current measured state $x_t$, the prior information, and the online collected data.
We denote the optimal solution of \eqref{sdp_onestep} at time $t$ by $\gamma_t^\star, H_t^\star, L_t^\star, \Gamma_t^\star$.
The optimal state-feedback gain is given by $F_t^\star=L_t^\star(H_t^\star)^{-1}$.

The following theorem shows that \eqref{sdp_onestep:obj}-\eqref{sdp_onestep:con3} minimizes an upper bound on the optimal cost of \eqref{mpc:onestep_obj}-\eqref{mpc:onestep_con2}.

\begin{mythm}
\label{theorem2}\upshape
Given a state $x_t\in\mathbb{R}^n$ at time $t$, suppose there exists $\gamma, H, L, \Gamma$ such that \eqref{sdp_onestep:con1}-\eqref{sdp_onestep:con3} is feasible. 
Let $P=\gamma H^{-1}$.
Then, the optimal cost of \eqref{mpc:onestep_obj}-\eqref{mpc:onestep_con2} is guaranteed to be at most $\|x_t\|_{P}^2$ and $\|x_t\|_{P}^2$ is upper bounded by $\gamma$, i.e.,
\[J^\star(x_t)\leq \|x_t\|_{P}^2\leq\gamma.\]
\end{mythm}
\begin{proof}
Applying the Schur complement to the constraint \eqref{sdp_onestep:con2} twice yields the following equivalent inequalities
\begin{subequations}\label{M1}
\begin{align}
    &\begin{bmatrix}
        -\gamma (P_p^\star)^{-1} \!\!\!\!\!\!&0\\
        0 \!\!\!\!\!\!&\begin{bmatrix}H\\L\end{bmatrix}\!\!(H-\frac{1}{\gamma}\Phi^\top\Phi)^{-1}\!\!\begin{bmatrix}H\\L\end{bmatrix}^\top\!
    \end{bmatrix}+
\Pi_t(\tau)+\Pi_p\prec 0,\label{M1:1}\\
&-H+\frac{1}{\gamma}\Phi^\top\Phi\prec  0.\label{M1:2}
\end{align}
\end{subequations}
According to  \eqref{assumption1:equation} and \eqref{S_t}, given any $\Gamma$ satisfying \eqref{sdp_onestep:con3}, the inequalities
\begin{subequations}\label{C1}
\begin{align}
&\begin{bmatrix}
I &A &B
\end{bmatrix}
\Pi_p
\begin{bmatrix}
I &A &B
\end{bmatrix}^\top\succeq 0,\\
&\begin{bmatrix}
I &A &B
\end{bmatrix}
\Pi_t(\tau)
\begin{bmatrix}
I &A &B
\end{bmatrix}^\top\succeq 0
\end{align}
\end{subequations}
hold for any $(A, B)\in S_t\cap \Sigma_p$.
Pre-multiplying \eqref{M1:1} with $\begin{bmatrix}I &A &B\end{bmatrix}$ and post-multiplying \eqref{M1:1} with $\begin{bmatrix}I &A &B\end{bmatrix}^\top$, the resulting inequality together with \eqref{C1} imply that the following inequality must hold for any $(A, B)\in S_t\cap \Sigma_p$ using the S-procedure \cite{scherer2000linear}
\begin{equation}\label{X1}\nonumber
    \begin{bmatrix}
        I \\
        A^\top\\
        B^\top
    \end{bmatrix}^\top
    \!\!\!
    \begin{bmatrix}\!
        -\gamma (P_p^\star)^{-1} \!\!\!\!&0\\
        0 \!\!\!\!&\begin{bmatrix}H\\L\end{bmatrix}\!(H-\frac{1}{\gamma}\Phi^\top\Phi)^{-1}\!\begin{bmatrix}H\\L\end{bmatrix}^\top\!
    \end{bmatrix}\!\!
    \begin{bmatrix}
        I\\
        A^\top\\
        B^\top
    \end{bmatrix}\!\!\prec 0.
\end{equation}
This is equivalent to 
\begin{equation}\label{th1:1}
    -\gamma (P_p^\star)^{-1}+(AH+BL)(H-\frac{1}{\gamma}\Phi^\top\Phi)^{-1}(AH+BL)^\top\!\!\prec 0.
\end{equation}
Using the Schur complement, \eqref{th1:1} together with \eqref{M1:2} is equivalent to 
\begin{equation}\label{th1:2}
    \begin{bmatrix}
        -H+\frac{1}{\gamma}\Phi^\top \Phi& (AH+BL)^\top\\
        AH+BL &-\gamma (P_p^\star)^{-1}
    \end{bmatrix}\prec  0.
\end{equation}
Using the Schur complement again, \eqref{th1:2} yields the equivalent inequality
\begin{subequations}\label{th1:3}
\begin{align}
    &(AH+BL)^\top \gamma^{-1}P_p^\star (AH+BL)\!-\!H\!+\!\frac{1}{\gamma}\Phi^\top \Phi\prec  0,\label{th1:31}\\
    & -\gamma P_p^\star\prec  0.\label{th1:32}
\end{align}
\end{subequations}
Let $P=\gamma H^{-1}$ and $F=LH^{-1}$.
Multiplying both sides of \eqref{th1:31} with $\sqrt{\gamma}H^{-1}$, we have
\begin{equation}\label{th1:4}
\begin{aligned}
    (A+BF)^\top P_p^\star (A+BF)-P
+Q+F^\top R F\prec  0.
\end{aligned}
\end{equation}
holds for any $(A, B)\in S_t\cap \Sigma_p$.
Multiplying \eqref{th1:4} by $\bar{x}_0(t)^\top$ and $\bar{x}_0(t)$ from left and right, respectively, the following inequality holds for any $(A, B)\in S_t\cap \Sigma_p$
\begin{equation}\label{th1:6}
\begin{aligned}
    \bar{x}_0(t)^\top \!(A+BF)^\top\! P_p^\star(A+BF)\bar{x}_0(t)-\bar{x}_0(t)^\top \!P \bar{x}_0(t)\\ \leq 
    -\bar{x}_0(t)^\top\! (Q+F^\top\! \!RF)\bar{x}_0(t).
\end{aligned}
\end{equation}
The inequality \eqref{th1:6} implies that the following inequality is satisfied  for the state $\bar{x}_0(t)$, input $\bar{u}_0(t)=F\bar{x}_0(t)$ and $\bar{x}_1(t)$ predicted by the system dynamics \eqref{mpc:onestep_con1} with any $(A, B)\in S_t\cap \Sigma_p$
\begin{equation}\label{Vconstraint}
\|\bar{x}_1(t)\|_{P_p^\star}^2-\|\bar{x}_0(t)\|_{P}^2\leq -\ell(\bar{u}_0(t), \bar{x}_0(t)).
\end{equation}
Since the inequality \eqref{Vconstraint} holds for any $(A, B)\in S_t\cap \Sigma_p$, it also holds for the worst-case value, i.e., we obtain
\begin{equation}\label{worstV}
    \max_{(A, B)\in S_t\cap \Sigma_p}\ell(u_t, x_t)+\|\bar{x}_1(t)\|_{P_p^\star}^2\leq \|\bar{x}_0(t)\|_{P}^2.
\end{equation}
This provides an upper bound on the optimal cost of \eqref{mpc_onestep}.
Since $\bar{x}_0(t)=x_t$ by \eqref{mpc:onestep_con2}, we have $J^\star(x_t)\leq \|x_t\|_P^2$.
Using the Schur complement, $\|x_t\|_{P}^2\leq \gamma$ is equivalent to the inequality \eqref{sdp_onestep:con1}.
In conclusion, given that \eqref{sdp_onestep:con1}-\eqref{sdp_onestep:con3} hold, we have shown that $\gamma$ is an upper bound on the optimal cost of \eqref{mpc:onestep_obj}-\eqref{mpc:onestep_con2}.
\end{proof}

\begin{remark}\upshape
Problem \eqref{sdp_onestep:obj}-\eqref{sdp_onestep:con3} provides an upper bound on the optimal cost of the problem \eqref{mpc:onestep_obj}-\eqref{mpc:onestep_con2} and a corresponding state-feedback control gain. 
Additionally, similar to the proof of Theorem~\ref{theorem1} (ii), the constraint in \eqref{sdp_onestep:con4} ensures that the constraint \eqref{mpc:onestep_con3} is satisfied for the
closed-loop system and the derived state-feedback gain.
The proof for deriving the upper bound on the worst-case cost is inspired by the data-driven min–max MPC approach in \cite{xie2024minmax}.
However, \cite{xie2024minmax} considers an infinite-horizon cost, whereas we consider a one-step stage cost combined with a terminal cost in the adaptive data-driven min-max MPC problem.
This difference leads to a different convex reformulation of the resulting SDP.
\end{remark}

We apply the adaptive data-driven min-max MPC scheme according to Algorithm~1.
At the initial time $t=0$, given the state $x_0$, we solve the optimization problem \eqref{sdp_initial} and obtain the optimal state-feedback gain $F_p^\star$ and the matrix $P_p^\star$. 
The computed input $u_0=F_p^\star x_0$ is implemented.
At the next sampling time $t+1$, we measure the state $x_{t+1}$.
With the collection of online input-state data, we characterize the set $S_t$ by Lemma~\ref{lemma1}. 
We add a new optimization variable $\tau_t$ to \eqref{sdp_onestep:con3} and update the constraint \eqref{sdp_onestep:con2} by incorporating the collected online input-state data
and variable $\tau_t$ into $\Pi_t(\tau)$.
Then we solve the problem \eqref{sdp_onestep} and iterate the above procedure at the next time step.

\begin{algorithm}
  \caption{Adaptive Data-driven Min-Max MPC Scheme.}\label{euclid}
  \begin{algorithmic}[1]
      \State At time $t=0$, measure state $x_0$
      \State Solve the problem \eqref{sdp_initial}
      \State Apply the input $u_0=F_p^\star x_0$
      \State Set $t=t+1$, measure state $x_t$
      \State Update the constraints \eqref{sdp_onestep:con2} and \eqref{sdp_onestep:con3}
      \State Solve the problem \eqref{sdp_onestep}
      \State Apply the input $u_t=F_t^\star x_t$
      \State Set $t=t+1$, measure state $x_t$, go back to 5
  \end{algorithmic}
\end{algorithm}

\begin{remark}\upshape
Algorithm~1 resolves problem~\eqref{sdp_onestep} at every time step. 
At each step, a new optimization variable is introduced into the problem formulation. 
To reduce computational complexity, old input-state data can be discarded, as discussed in Remark~4.
In addition, the repeated solution of problem~\eqref{sdp_onestep} can be stopped once the state converges sufficiently close to the origin. 
In this case, the optimal state-feedback gain $F_p^\star$ can be applied instead to maintain stability while reducing computational effort.
% In Algorithm~2, as more data are collected, the constraints \eqref{lemma2_e2} and \eqref{sdp_onestep_noise:con2} must be updated at each time step.
% As the length of the data sequence increases, the computational complexity grows due to the increasing number of constraints and optimization variables.
% To reduce the computational burden, one may instead use only the most recent input–state data, thereby keeping the number of constraints and optimization variables constant.
\end{remark}

\subsection{Closed-loop Guarantees}\label{sec:4.3}
In the following theorem, we prove recursive feasibility of the problem \eqref{sdp_onestep}.
Then, we show exponential stability and constraint satisfaction of the closed-loop system resulting from the proposed scheme.

\begin{mythm}\label{theorem3}\upshape
Suppose Assumptions \ref{assumption1} and \ref{assumption2} hold, and the problem \eqref{sdp_initial} is feasible for a given initial state $x_0\in\mathbb{R}^n$. Then
\begin{enumerate}[(i)]
\item the optimization problem \eqref{sdp_onestep} is feasible at any time $t\in\mathbb{N}$;
\item the equilibrium $x^s=0$ is exponentially stable for the closed-loop system;
\item the closed-loop trajectory resulting from Algorithm~1 satisfies the constraints, i.e., $\|C_x x_t+C_u u_t\|\leq 1$ for all $t\in\mathbb{N}$.
\end{enumerate}
\end{mythm}
\begin{proof}
The proof consists of two parts. In Part I, we prove recursive feasibility of the problem \eqref{sdp_onestep}.
In Part II, we prove that the equilibrium $x_s=0$ is exponentially stable for the closed-loop system.
The proof for constraint satisfaction is similar to Theorem~1 and, hence, we omit the proof.

\textbf{Part I: } 
Suppose that the optimization problem \eqref{sdp_initial} is feasible at time $t=0$ with initial state $x_0$.
To prove recursive feasibility, we define a candidate solution of problem \eqref{sdp_onestep} at time $t\in\mathbb{N}$ as follows:
\begin{equation}\label{th3:p11}
    \gamma_t=\gamma_p^\star, H_t=H_p^\star, L_t=L_p^\star, \Gamma_t=(
    \tau_p^\star, 0, \ldots, 0).
\end{equation}
Constraints \eqref{sdp_onestep:con2}-\eqref{sdp_onestep:con4} are trivially satisfied with the defined candidate solution.
We need to prove that the constraint \eqref{sdp_onestep:con1} is feasible with this candidate solution for the future states $x_{t+1}, t\in\mathbb{N}$.

\textbf{Case I (time $t=1$): }%Now we show that the constraint \eqref{sdp_onestep:con1} is feasible with the candidate solution \eqref{th3:p11} at time $t=1$, and we address $t\geq 2$ afterwards.
At time $t=0$, since the constraint \eqref{sdp_initial:con2} is satisfied for $\gamma_p^\star, H_p^\star, L_p^\star, \tau_p^\star$, we have
\begin{equation}
    (A+BF_p^\star)^\top P_p^\star(A+BF_p^\star)-P_p^\star\prec -(Q+{F_p^\star}^\top RF_p^\star)\label{thm3:p1}
\end{equation}
holds for any $(A, B)\in\Sigma_{p}$, compare the proof of \cite[Theorem 1]{xie2024minmax}. 
Multiplying \eqref{thm3:p1} from the left and right by $x_0^\top$ and $x_0$, respectively, we obtain
\begin{equation}
\begin{aligned}
    x_0^\top (A+BF_p^\star)^\top P_p^\star(A+BF_p^\star)x_0-x_0^\top P_p^\star x_0\leq \\
    -x_0^\top(Q+{F_p^\star}^\top RF_p^\star)x_0\label{lemma1:proof1}
\end{aligned}
\end{equation}
holds for any $(A, B)\in\Sigma_{p}$.
Since $(A_0, B_0)\in \Sigma_p$,  $x_1=(A_0+B_0F_p^\star)x_0$ and \eqref{lemma1:proof1} holds for any $(A, B)\in\Sigma_{p}$, we have
\begin{equation}\label{th3:p2}
    x_1^\top P_p^\star x_1-x_0^\top P_p^\star x_0\leq -x_0^\top(Q+{F_p^\star}^\top RF_p^\star)x_0.
\end{equation}
Since \eqref{sdp_initial:con1} is satisfied for $H_p^\star$ and $P_p^\star=\gamma_p^\star (H_p^\star)^{-1}$, using the Schur complement, we have $x_0^\top P_p^\star x_0\leq \gamma_p^\star$, which further gives  
\begin{equation}\nonumber
    x_1^\top P_p^\star x_{1}\overset{\eqref{th3:p2}}{\leq} x_0^\top P_p^\star x_0\leq \gamma_p^\star.
\end{equation}
Using $x_1^\top P_p^\star x_1\leq \gamma_p^\star$, $P_p^\star=\gamma_p^\star (H_p^\star)^{-1}$ and the Schur complement,  the inequality \eqref{sdp_onestep:con1} is feasible with $H_p^\star$ at time $t=1$ given the state $x_1$.
Thus, the problem \eqref{sdp_onestep} is feasible with the candidate solution \eqref{th3:p11} at time $t=1$.

\textbf{Case II (time $t\geq 2$): }%Now we show that the constraint \eqref{sdp_onestep:con1} is feasible with the candidate solution \eqref{th3:p11} at time $t\geq 2$.
At time $t=1$, since the constraints \eqref{sdp_onestep:con2} and  \eqref{sdp_onestep:con3} are satisfied for $\gamma_1^\star, H_1^\star, L_1^\star, \Gamma_1^\star$, we have 
\begin{equation}\label{thm3:p4}
    (A+BF_1^\star)^\top P_p^\star (A+BF_1^\star)-P_1^\star\prec -(Q+{F_1^\star}^\top R F_1^\star)
\end{equation}
holds for any $(A, B)\in S_1\cap \Sigma_p$ by \eqref{th1:4}.
Multiplying \eqref{thm3:p4} by $x_1^\top$ and $x_1$ from the left and right, we have
\begin{equation}
\begin{aligned}
    x_1^\top (A+BF_1^\star)^\top P_p^\star(A+BF_1^\star)x_1-x_1^\top P_1^\star x_1\leq \\
    -x_1^\top(Q+{F_1^\star}^\top RF_1^\star)x_1\label{thm3:p5}
\end{aligned}
\end{equation}
holds for any $(A, B)\in S_1\cap \Sigma_p$.
Since $(A_1, B_1)\in S_1\cap \Sigma_p$, $x_2=(A_1+B_1 F_1^\star)x_1$ and \eqref{thm3:p5} holds for any $(A, B)\in S_1\cap \Sigma_p$, we have
\begin{equation}\label{th3:p6}
    x_2^\top P_p^\star x_2-x_1^\top P_1^\star x_1\leq -x_1^\top(Q+{F_1^\star}^\top RF_1^\star)x_1.
\end{equation}
Since $\gamma_1^\star$ is the optimal solution of the problem \eqref{sdp_onestep} at time $t=1$, while $\gamma_p^\star$ is a feasible solution, we have $\gamma_1^\star\leq \gamma_p^\star$.
Hence, the following inequalities hold
\begin{equation}\label{thm3:111}\nonumber
    x_2^\top P_p^\star x_2\overset{\eqref{th3:p6}}{\leq} x_1^\top P_1^\star x_1\leq \gamma_1^\star\leq \gamma_p^\star.
\end{equation}
Thus, the inequality \eqref{sdp_onestep:con1} is feasible with $H_p^\star$ given the state $x_2$ and the problem \eqref{sdp_onestep} is feasible with the candidate solution \eqref{th3:p11} at time $t=2$.
Applying this argument inductively, we prove that the problem \eqref{sdp_onestep} is feasible with the candidate solution \eqref{th3:p11} for time $t\in\mathbb{I}_{[3, \infty)}$.

\textbf{Part II: }We choose $V(x_t)=\|x_t\|_{P_p^\star}^2$ as the Lyapunov function candidate.
%As shown in Theorem~\ref{theorem2}, $\|x_t\|_P^2$ with any feasible solution of the LMIs \eqref{sdp_onestep:con1}-\eqref{sdp_onestep:con2} is an upper
%bound on the optimal cost of \eqref{mpc:onestep_obj}-\eqref{mpc:onestep_con2}. Thus,  we have $\|x_t\|_{P_p^\star}^2\geq\|x_t\|_{P_t^\star}^2\geq  \ell(u_t, x_t) \geq  \|x_t\|^2_Q$. 
%Moreover, as shown in Part I of the proof, \eqref{th3:p11} is a candidate solution, thus we have $\|x_t\|_{P_p^\star}^2\geq \|x_t\|_{P_t^\star}^2$ by optimality.
%The Lyapunov function $V(x_t)$ is lower bounded by $\|x_t\|_Q^2$ and upper bounded by ...
Since \eqref{th1:4} holds for $P=P_t^\star$ and $F=F_t^\star$, 
\begin{equation}\label{thm3:p7}
    (A+BF_t^\star)^\top P_p^\star (A+BF_t^\star)-P_t^\star\prec -(Q+{F_t^\star}^\top R F_t^\star)
\end{equation}
holds for any $(A, B)\in S_t\cap\Sigma_p$.
Multiplying \eqref{thm3:p7} by $x_t^\top$ and $x_t$ from left and right, and since $(A_t, B_t)\in S_t\cap \Sigma_p$ and $x_{t+1}=(A_t+B_t F_t^\star)x_t$, we have
\begin{equation}\label{thm3:stab}
    x_{t+1}^\top P_p^\star x_{t+1}-x_t^\top P_t^\star x_t\leq -x_t^\top(Q+{F_t^\star}^\top RF_t^\star)x_t.
\end{equation}
Since $x_t^\top P_t^\star x_t\leq x_t^\top P_p^\star x_t$ by optimality and using \eqref{thm3:stab}, we have
\begin{equation}\label{thm3:stab1}
    x_{t+1}^\top P_p^\star x_{t+1}-x_t^\top P_p^\star x_t\leq -x_t^\top(Q+{F_t^\star}^\top RF_t^\star)x_t.
\end{equation}
Since $x_t^\top(Q+{F_t^\star}^\top RF_t^\star)x_t\geq \lambda_{\min}(Q)\|x_t\|^2$ and \eqref{thm3:stab1} holds,
\begin{equation}\nonumber
    x_{t+1}^\top P_p^\star x_{t+1}-x_t^\top P_p^\star x_t\leq -\lambda_{\min}(Q)\|x_t\|^2.
\end{equation}
This shows that $x^s=0$ is exponentially stable for the closed loop.
\end{proof}

\begin{remark}\upshape
Theorem~\ref{theorem3} shows that if the optimization problem~\eqref{sdp_initial} is feasible, then the optimization problem \eqref{sdp_onestep} is recursively feasible.
The idea for proving recursive feasibility is to use the optimal solution of the problem \eqref{sdp_initial} as a candidate solution for problem \eqref{sdp_onestep}.
The state-feedback control law derived from the optimal solution of the problem \eqref{sdp_initial} could be viewed as a backup controller to ensure closed-loop stability and feasibility of the problem \eqref{sdp_onestep}.
The idea for proving exponential stability is to construct a Lyapunov function $V(x_t)=\|x_t\|_{P_p^\star}^2$. 
By using the fact that $P_t^\star$ is an optimal solution of problem \eqref{sdp_onestep} at time $t$, while $P_p^\star$ is only a candidate solution, we can show that the decay of the Lyapunov function is negative, which proves exponential stability of the closed-loop system.
\end{remark}

\begin{remark}\label{theorem_infeasible}\upshape
Theorem~\ref{theorem3} relies on the assumption that the SDP problem \eqref{sdp_initial} is feasible at the initial time, which implies that a stabilizing state-feedback controller can be designed using only the prior information on the system matrices.
The online data is used to further improve closed-loop performance.
However, in some cases, the available prior information may be not informative enough for controller synthesis or the prior parameter uncertainty set is too big, i.e., the SDP problem~\eqref{sdp_initial} is initially infeasible.
A possible method in these cases is to apply randomly chosen inputs during the first few initial time steps to collect online data. 
The collected online data provides information for the system dynamics at each current time step.
By combining online data with the prior information, it may then be possible to design a stabilizing controller. 
The resulting SDP problem is similar to~\eqref{sdp_onestep}, except that the term $-\gamma (P_p^\star)^{-1}$ in constraint~\eqref{sdp_onestep:con2} is replaced by $-H$. 
In this setting, recursive feasibility of the SDP problem cannot be guaranteed. Nevertheless, in some cases, a stabilizing controller can still be obtained if the resulting SDP problem is feasible at all time steps.
We will demonstrate this in the simulation results.
\end{remark}

\section{Adaptive data-driven min-max MPC scheme for noisy systems}\label{sec:5}

In this section, we consider an unknown discrete-time LTV system with process noise
\begin{equation}\label{system_noise}
    x_{t+1}=A_tx_t+B_tu_t+\omega_t,
\end{equation}
where $\omega_t\in\mathbb{R}^n$ denotes the unknown process noise.
All other notations and assumptions remain the same as those for system \eqref{system}, i.e., Assumption \ref{assumption1} and \ref{assumption2} hold.
Further, we assume that the process noise $\omega_t$ satisfies the following assumption.

\begin{assum}\label{assumption3}\upshape
For all $t\in\mathbb{N}$, the noise $\omega_t\in\mathbb{R}^n$ satisfied $\|\omega_t\|_G\leq 1$ for a known matrix $G\succ 0$.
\end{assum}

We characterize the set of consistent system matrices for the system \eqref{system_noise} at time $t$ using the following sequence of input-state data collected from the noisy system \eqref{system_noise}
\begin{subequations}\label{data_noise}
    \begin{align}
        \tilde{X}_{t}&=\begin{bmatrix}
            x_0 &x_1 &\ldots &x_{t-1} &x_t
        \end{bmatrix},\\
        \tilde{U}_{t}&=\begin{bmatrix}
            u_0 &u_1 &\ldots &u_{t-1}
        \end{bmatrix}.
    \end{align}
\end{subequations}

%The system dynamics at time $t-i$ can be written as
%\begin{equation}\nonumber
    %x_{t-i+1}=A_{t}x_{t-i}+B_{t}u_{t-i}+\tilde{\omega}_t^{t-i},
%\end{equation}
%where $\tilde{\omega}^{t-i}_t=\Delta A^{t-i}_t x_{t-i}+\Delta B^{t-i}_t u_{t-i}+\omega_t=\omega^{t-i}_t+\omega_t$ is the virtual disturbance.
The set of $(A, B)$ at time $t$ that are consistent with the sequence of input-state data \eqref{data_noise} is defined as 
\[\tilde{S}_t=\{(A, B):(A, B)\in\tilde{\Sigma}_{t,i}, \forall i\in\mathbb{I}_{[1, t]}\},\]
where 
\[\tilde{\Sigma}_{t, i}\!=\!\left\{(A, B): 
\begin{gathered}
    \exists (\Delta A^{t-i}_t, \Delta B^{t-i}_t)\in\Pi_i\text{ and }\\
   \omega_t\text{ satisfying Assumption~\ref{assumption3}}\text{ s.t. }\\x_{t-i+1}=Ax_{t-i}+Bu_{t-i}+\\
    \Delta A^{t-i}_t x_{t-i}+\Delta B^{t-i}_t u_{t-i}+\omega_t
\text{ holds}
\end{gathered}\right\}.\]

We derive a data-driven characterization of the set $\tilde{S}_t$ based on the data \eqref{data_noise} in the following lemma.

\begin{mylem}\upshape\label{lemma2}
Suppose Assumption~\ref{assumption2} and ~\ref{assumption3} hold, and $\begin{bmatrix}x_{t-i}\\ u_{t-i}\end{bmatrix}\neq 0$ for all $i\in\mathbb{I}_{[1, t]}$.
Then, the set $\tilde{S}_t$ is equal to
\begin{equation}\label{tilde_S_t}
\left\{\!(A, B)\!:
\begin{gathered}
\begin{bmatrix}
I \!\!&A \!\!&B
\end{bmatrix}
\Pi_{t}(O)
\begin{bmatrix}
I \!\!&A \!\!&B
\end{bmatrix}^\top\!\succeq\! 0\\
\forall O=(O_1, \ldots, O_t), O_i \text{ satisfying }\eqref{lemma2_e2}, i\in\mathbb{I}_{[1, t]}
\end{gathered}\!
\right\}\!,
\end{equation}
where 
\begin{equation}\nonumber
\Pi_{t}(O)\!=\!\sum_{i=1}^{t}
\begin{bmatrix}
I &x_{t-i+1}\\
0 &-x_{t-i}\\
0 &-u_{t-i}
\end{bmatrix}
O_i
\begin{bmatrix}
I &x_{t-i+1}\\
0 &-x_{t-i}\\
0 &-u_{t-i}
\end{bmatrix}^\top,
\end{equation}
$O_i=\begin{bmatrix}O_{i, 11} &O_{i, 12}\\O_{i, 21} &O_{i, 22}\end{bmatrix}$, and the following inequalities hold with $N_i$ defined as in \eqref{Ni}
\begin{equation}\label{lemma2_e2}
\begin{aligned}
&\begin{bmatrix}
        O_{i, 11} \!\!&0 \!\!&0 \!\!&O_{i, 12}\\
        0 \!\!&0 \!\!&0 \!\!&0\\
        0 \!\!&0 \!\!&0 \!\!&0\\
        O_{i, 21} \!\!&0 \!\!&0 \!\!&O_{i, 22}
    \end{bmatrix}\!\!-\!\!\lambda_{i, 1}\!
    \begin{bmatrix}
    N_i \!\!&0\\0 \!\!&0
    \end{bmatrix}\!\!-\!\!\lambda_{i, 2}\!\begin{bmatrix}
        G^{-1} \!\!\!&0 \!\!\!&0 \!\!\!&0\\
        0 \!\!\!&0 \!\!\!&0 \!\!\!&0\\
        0 \!\!\!&0 \!\!\!&-I \!\!\!&0\\
        0 \!\!\!&0 \!\!\!&0 \!\!\!&0
    \end{bmatrix}\!\succeq 0,\\
&\text{for some }\lambda_{i, 1}\geq 0, \lambda_{i, 2}\geq 0, \forall i\in\mathbb{I}_{[1, t]}.
\end{aligned}
\end{equation}
\end{mylem}
\begin{proof}
Using the derivation in the proof of Lemma~\ref{lemma1}, we have $\omega_t^{t-i}=\Delta A^{t-i}_t x_{t-i}+\Delta B^{t-i}_t u_{t-i}$ satisfying
\begin{equation}\label{l2p1}
    \begin{bmatrix}
        I&{\omega^{t-i}_t}                    
    \end{bmatrix} N_i
    \begin{bmatrix}
        I&{\omega^{t-i}_t}                    
    \end{bmatrix}^\top \succeq 0.
\end{equation}
Using Assumption~\ref{assumption3}, we have
\begin{equation}\label{l2p2}
    \begin{bmatrix}
        I&\omega_t                    
    \end{bmatrix} \begin{bmatrix}
        G^{-1} &0 \\
        0 &-I
    \end{bmatrix}
    \begin{bmatrix}
        I&{\omega_t}                    
    \end{bmatrix}^\top \succeq 0.
\end{equation}
We want to find a set that contains the set \[\mathcal{A}_t^{t-i}=\{\tilde{\omega}_t^{t-i}:\omega_t^{t-i} \text{ satisfying }\eqref{l2p1}, \omega_t\text{ satisfying }\eqref{l2p2}\}.\]
The set $\{\tilde{\omega}_t^{t-i}:\begin{bmatrix}
        I&{\tilde{\omega}^{t-i}_t}                 
    \end{bmatrix} \tilde{N}_i
    \begin{bmatrix}
        I&{\tilde{\omega}^{t-i}_t}                
    \end{bmatrix}^\top \succeq 0\}$ with $\tilde{N}_i=\begin{bmatrix}\tilde{N}_{i, 11} &\tilde{N}_{i, 12}\\
\tilde{N}_{i, 21} &\tilde{N}_{i, 22}\end{bmatrix}$ contains the set $\mathcal{A}_t^{t-i}$
if and only if for for every $\omega_t^{t-i}$  satisfying \eqref{l2p1}, $\omega_t$ satisfying \eqref{l2p2}, we have $\begin{bmatrix}
        I&{\tilde{\omega}^{t-i}_t}                 
    \end{bmatrix} \tilde{N}_i
    \begin{bmatrix}
        I&{\tilde{\omega}^{t-i}_t}                
    \end{bmatrix}^\top \succeq 0$.
Using the S-procedure, this condition is true if there exist $\alpha_{i, 1}, \alpha_{i, 2}\geq 0$ such that \eqref{l2p4} holds for every $\omega_t^{t-i}, \omega_t$.
\begin{figure*}
    \begin{equation}\label{l2p4}
        \begin{bmatrix}
            I &\omega_t^{t-i} &\omega_t &\tilde{\omega}_t^{t-i}
        \end{bmatrix}
        \underbrace{\left (\begin{bmatrix}
        \tilde{N}_{i, 11} \!\!&0 \!\!&0 \!\!&\tilde{N}_{i, 12}\\
        0 \!\!&0 \!\!&0 \!\!&0\\
        0 \!\!&0 \!\!&0 \!\!&0\\
        \tilde{N}_{i, 21} \!\!&0 \!\!&0 \!\!&\tilde{N}_{i, 22}
    \end{bmatrix}\!\!-\!\alpha_{i, 1}\!
    \begin{bmatrix}
    N_i \!\!&0\\0 \!\!&0
    \end{bmatrix}\!\!-\!\alpha_{i, 2}\!\begin{bmatrix}
        G^{-1} \!\!\!&0 \!\!\!&0 \!\!\!&0\\
        0 \!\!\!&0 \!\!\!&0 \!\!\!&0\\
        0 \!\!\!&0 \!\!\!&-I \!\!\!&0\\
        0 \!\!\!&0 \!\!\!&0 \!\!\!&0
    \end{bmatrix}\right )}_S
    \begin{bmatrix}
            I &\omega_t^{t-i} &\omega_t &\tilde{\omega}_t^{t-i}
        \end{bmatrix}^\top\!\succeq 0
    \end{equation}
\end{figure*}
The condition \eqref{l2p4} is equivalent to the existence of $\alpha_{i, 1}, \alpha_{i, 2}\geq 0$ such that $S\succeq 0$.
Using the same arguments as the proof of Lemma~\ref{lemma1}, we have
\begin{equation}\label{l2p6}
\tilde{\Sigma}_{t, i}\!\!=\!\!\left\{\!\!(A, B)\!\!:\!\!\!
\begin{bmatrix}
I\\ A^\top \\B^\top
\end{bmatrix}^\top\!\!\!
\begin{bmatrix}
I \!\!\!&x_{t-i+1}\\
0 \!\!\!&-x_{t-i}\\
0 \!\!\!&-u_{t-i}
\end{bmatrix}
\!\!\tilde{N}_i\!\!
\begin{bmatrix}
I \!\!\!&x_{t-i+1}\\
0 \!\!\!&-x_{t-i}\\
0 \!\!\!&-u_{t-i}
\end{bmatrix}^\top\!\!\!
\begin{bmatrix}
I\\ A^\top \\B^\top
\end{bmatrix}
\!\!\!\succeq \!0\!
\right\}
\end{equation}
for any $\tilde{N}_i$ such that $S\succeq 0$ holds.
Using the characterization of $\tilde{\Sigma}_{t, i}$ as in \eqref{l2p6}, we obtain the characterization of the set $\tilde{S}_t$ as in \eqref{tilde_S_t} with $O_i=\tau_i\tilde{N}_i$, $\lambda_{i, 1}=\tau_i\alpha_{i, 1}$, $\lambda_{i, 2}=\tau_i\alpha_{i, 2}$, analogous to \cite{berberich2023combining, bisoffi2021trade}.
%$\hfill\blacksquare$
\end{proof}

\begin{remark}\upshape
The idea of the data-driven characterization of the set $\tilde{S}_t$ is to first derive a QMI for $\tilde{\omega}^{t-i}_t$ using the bounds on the process noise and on the variation of the system matrices.
Using this QMI, the set of consistent system matrices can be characterized in a way analogous to the approach in Lemma~\ref{lemma1}.
Note that the matrix $O_i$ is not computed explicitly in this lemma.
It is later treated as an optimization variable in the corresponding SDP.
\end{remark}

The data-driven min-max MPC optimization problem using prior knowledge for LTV system with process noise is formulated analogous to \eqref{mpc_initial}.
In particular, the nominal system dynamics without process noise is used to predict the future state.
Next, we formulate an SDP problem that accounts for the influence of the process noise.
The idea of the SDP reformulation is related to the approach in \cite{xie2026data} for LTI systems, which is extended to LTV systems in the present paper.

At time $t=0$, given prior knowledge of the system dynamics \eqref{system_noise} in Assumption~\ref{assumption1}, an initial state $x_0\in\mathbb{R}^n$ and a constant $c>\lambda_{\min}(Q)$, the SDP is formulated as follows:
\begin{subequations}\label{sdp_initial_noise}
\begin{align}
    &\minimize\limits_{\gamma>0, H\in\mathbb{R}^{n\times n}, L\in\mathbb{R}^{m\times n}, \tau_p\geq 0}\gamma\label{sdp_initial_noise:obj}\\
    \text{s.t. }
    &\eqref{sdp_initial:con1}, \eqref{sdp_initial:con3}\text{ hold},\\
    &\begin{bmatrix}
        \begin{bmatrix}
            -H+\frac{\gamma}{c} \!\!&0\\
            0 \!\!&0
        \end{bmatrix}\!+\Pi_{p} &
        \begin{bmatrix}
            0\\
            H\\
            L
        \end{bmatrix}
        \!\!& 0\\
        \begin{bmatrix}
            0 &H &L^\top
        \end{bmatrix} &-H \!\!&\Phi^\top\\
        0 &\Phi \!\!& -\gamma I
    \end{bmatrix}\prec 0.\label{sdp_initial_noise:con2}
\end{align}
\end{subequations}
We denote the optimal solution of \eqref{sdp_initial_noise} by $\gamma_p^\star, H_p^\star, L_p^\star, \tau_p^\star$.
The corresponding optimal state-feedback gain is given by $F_p^\star=L_p^\star(H_p^\star)^{-1}$ and we define $P_p^\star=\gamma_p^\star (H_p^\star)^{-1}$.

The constant $c$ is required to establish robust stability of the resulting closed-loop system  by applying the state-feedback gain $F_p^\star$.
The method how to choose $c$ is explained in \cite{xie2026data}.
Similar to the LTI system results in \cite[Theorem 1 and 2]{xie2026data}, problem \eqref{sdp_initial_noise} provides an upper bound on the optimal cost of the problem \eqref{mpc_initial} and a corresponding state-feedback control gain.
The resulting closed-loop system converges exponentially to a RPI set and satisfies the constraint.

We present an adaptive data-driven min-max MPC scheme for LTV systems with process noise using both the prior knowledge and online input-state data to improve closed-loop performance.
Given $\tilde{S}_t\cap\Sigma_p$ and the current state $x_t\in\mathbb{R}^n$ at time $t$, the adaptive data-driven min-max MPC optimization problem is formulated in the same manner as \eqref{mpc_onestep}, except that the maximazation is now over $(A, B)\in\tilde{S}_t\cap\Sigma_p$.
Since this problem is computationally intractable, we formulate an SDP to compute an upper bound on the optimal cost of the adaptive data-driven min-max MPC optimization problem and a corresponding state-feedback control law.

At time $t$, given $\tilde{S}_t\cap \Sigma_p$, the matrix $P_p^\star$ and the current state $x_t\in\mathbb{R}^n$, the SDP is formulated as follows:
\begin{subequations}\label{sdp_onestep_noise}
\begin{align}
    &\minimize\limits_{\gamma>0, H\in\mathbb{R}^{n\times n}, L\in\mathbb{R}^{m\times n}, \tau_p\geq 0, \atop \lambda_{i, 1}, \lambda_{i, 2}, O_i\in\mathbb{R}^{(n+1)\times (n+1)}, i\in\mathbb{I}_{[1, t]}}\gamma\label{sdp_onestep_noise:obj}\\
    \text{s.t. }
    &\eqref{sdp_onestep:con1}, \eqref{sdp_onestep:con4}\text{ hold}, \eqref{lemma2_e2} \text{ hold } \forall i\in\mathbb{I}_{[1, t]}, \label{sdp_onestep_noise:con1}\\
    &\begin{bmatrix}
        \begin{bmatrix}
            -\gamma (P_p^\star)^{-1}+\frac{\gamma}{c} \!\!\!&0\\
            0 \!\!\!&0
        \end{bmatrix}\!+\!\Pi_{p}\!+\!\Pi_t(O) &
        \begin{bmatrix}
            0\\
            H\\
            L
        \end{bmatrix}
        & 0\\
        \begin{bmatrix}
            0 &H &L^\top
        \end{bmatrix} &-H &\Phi^\top\\
        0 &\Phi & -\gamma I
    \end{bmatrix}\prec 0\label{sdp_onestep_noise:con2}.
\end{align}
\end{subequations}
We denote the optimal solution of \eqref{sdp_onestep} at time $t$ by $\gamma_t^\star, H_t^\star, L_t^\star, \Gamma_t^\star$.
The optimal state-feedback gain is given by $F_t^\star=L_t^\star(H_t^\star)^{-1}$.

In the following theorem, we show that problem \eqref{sdp_onestep_noise} provides an upper bound on the optimal cost of the adaptive data-driven min-max MPC problem.

\begin{mythm}
\label{theorem4}\upshape
Given a state $x_t\in\mathbb{R}^n$ at time $t$, suppose there exists $\gamma, H, L, \tau_p, \lambda_{i, 1}, \lambda_{i, 2}, O_i, i\in\mathbb{I}_{[1, t]}$ such that \eqref{sdp_onestep_noise:con1}-\eqref{sdp_onestep_noise:con2} is feasible. 
Let $P=\gamma H^{-1}$.
Then, the optimal cost of the adaptive data-driven min-max MPC problem is guaranteed to be at most $\|x_t\|_{P}^2$ and $\|x_t\|_{P}^2$ is upper bounded by $\gamma$, i.e.,
\[J^\star(x_t)\leq \|x_t\|_{P}^2\leq\gamma.\]
\end{mythm}
\begin{proof}
The proof is similar to \cite[Theorem~1]{xie2026data}, so we only provide a short sketch here. 
Following the arguments of \cite[inequalities (8)-(18)]{xie2026data}, we could show that \eqref{sdp_onestep_noise:con2}, \eqref{lemma2_e2} and $\tau_p\geq 0$ are equivalent to the following inequalities
\begin{equation}\label{thm3_1}
 \begin{bmatrix}
    \!(A\!+\!BF)^\top \!P_p^\star (A\!+\!BF)\!-\!P\!\!+\!\!Q\!\!+\!\!F^\top\! R F  \!\!&(A\!+\!BF)^\top P_p^\star\\ P_p^\star(A\!+\!BF) \!\!&P_p^\star-cI
    \end{bmatrix}\!\!\prec\!  0.
\end{equation}
Then, following the arguments in \eqref{th1:4}-\eqref{worstV}, we prove that $\gamma$ is an upper bound on the adaptive data-driven min-max MPC problem.
\end{proof}

We apply the adaptive data-driven min-max MPC scheme for LTV systems with process noise according to Algorithm~2.
At the initial time $t=0$, given the state $x_0$, we solve the optimization problem \eqref{sdp_initial_noise} and obtain the optimal state-feedback gain $F_p^\star$ and the matrix $P_p^\star$. 
The computed input $u_0=F_p^\star x_0$ is implemented.
At the next sampling time $t+1$, we measure the state $x_{t+1}$.
We add new optimization variables $O_t, \lambda_{t, 1}, \lambda_{t, 2}$ and update $\Pi_t(O)$ in the constraint \eqref{sdp_onestep_noise:con2} and the constraint \eqref{lemma2_e2} using the collected online input-state data.
Then we solve the problem \eqref{sdp_onestep_noise} and iterate the above procedure at the next time step.
% At the initial time $t=0$, we solve the optimization problem \eqref{sdp_initial_noise} and implement only the first computed input $u_0=F_{p}^\star x_0$.
% At the next time step, we measure the state,  update the constraints \eqref{lemma2_e2} and \eqref{sdp_onestep_noise:con2}, and solve the optimization problem \eqref{sdp_onestep_noise}.
% We define $c_{RPI}=\frac{c^2}{\lambda_{\min}(Q)\lambda_{\min}(G)}$.
% If $\|x_t\|_{P_p^\star}^2>c_{RPI}$, we implement the first computed input $u_t=F_{t}^\star x_t$ and repeat this procedure.
% Upon convergence to the RPI set described by $\|x_t\|_{P_p^\star}^2\leq c_{RPI}$, then we define  $\tilde{F}=F_p^\star$ and $\tilde{P} = P_p^\star$, and denote the corresponding time instant by $\tilde{t}$.
% From this time onward, we stop solving the problem \eqref{sdp_onestep_noise} and directly apply $F_p^\star$ to the system.

\begin{algorithm}[htb]
\begin{algorithmic}[1]
\caption{\!Adaptive Data-driven Min-Max MPC Scheme for LTV Systems with Process Noise.\!\!\!}
    %\State \algorithmicrequire{ $U_f, X_f$, $Q, R, S_x, S_u$, $c$, $G$}\;
    \State At time $t=0$, measure state $x_0$\;
    \State Solve the problem \eqref{sdp_initial_noise}\;
    \State Apply the input $u_0=F_{p}^\star x_0$\;
    \State Set $t=t+1$,  measure state $x_t$\;
    \State Update the constraints \eqref{lemma2_e2} and \eqref{sdp_onestep_noise:con2}\;
    \State Solve the problem \eqref{sdp_onestep_noise}\;
    %\If{$\|x_t\|_{P_p^\star}^2>c_{RPI}$}\;
    \State Apply the input $u_t=F_{t}^\star x_t$\;
    \State Set $t=t+1$,  measure state $x_t$, go back to 5\;
    %\ElsIf{$\|x_t\|_{P_p^\star}^2\leq c_{RPI}$}
    %\State Apply the input $u_t=F_p^\star x_t$\;
    %\State Set $t=t+1$,  measure state $x_t$, go back to 12\;
    %\EndIf
    \label{algorithm:robustnoise}
\end{algorithmic}
\end{algorithm}

In the following theorem, we prove recursive feasibility of the problem \eqref{sdp_onestep_noise}, robust stability and constraint satisfaction for the resulting closed-loop system.

\begin{mythm}\label{theorem5}\upshape
Suppose Assumptions~\ref{assumption1}, \ref{assumption2}, and \ref{assumption3} hold,  the problem \eqref{sdp_initial_noise} is feasible for a given initial state $x_0\in\mathbb{R}^n$, 
and $\|x_0\|_{P_p^\star}^2\geq c_{RPI}$ with $c_{RPI}=\frac{c^2}{\lambda_{\min}(Q)\lambda_{\min}(G)}$. Then,
\begin{enumerate}[(i)]
\item the optimization problem is feasible at any time $t\in\mathbb{N}$;
\item the closed-loop system resulting from Algorithm~2 tends exponentially to the set $\mathcal{E}_{RPI}=\{x\in\mathbb{R}^n:\|x\|_{P_p^\star}^2\leq c_{RPI}\}$, i.e., 
\[\|x_{t+1}\|_{P_p^\star}^2-c_{RPI}\leq \beta (\|x_{t}\|_{P_p^\star}^2-c_{RPI})\]
for some $0<\beta<1$, and
stays inside the set $\mathcal{E}_{RPI}$ once the system enters it;
\item the closed-loop trajectory resulting from Algorithm~2 satisfies the constraints, i.e., $\|C_xx_t+C_u u_t\|\leq 1$ for all $t\in \mathbb{N}$.
\end{enumerate}
\end{mythm}
\begin{proof}
Part I proves recursive feasibility of the problem \eqref{sdp_initial_noise}.
Part II shows robust stability of the set $\mathcal{E}_{RPI}$.
The proof for constraint satisfaction is similar to Theorem~1 and, hence, we omit the proof.
Using the same argument as [inequalities (25)-(26)]\cite{xie2026data}, the upper bound and lower bound of $\|x\|_{P_p^\star}^2$ are $\frac{c}{\lambda_{\min}(Q)}\|x\|_Q^2$ and $\|x\|_Q^2$, respectively.

\textbf{Part I}:
Suppose that the optimization problem \eqref{sdp_initial_noise} is feasible at time $t=0$ with initial state $x_0$.
To prove recursive feasibility, we define a candidate solution of problem \eqref{sdp_onestep_noise} at time $t\in\mathbb{N}$ as follows:
\begin{equation}\label{candi2}
\begin{aligned}
    \gamma_t=\gamma_p^\star, H_t=H_p^\star, L_t=L_p^\star, \tau_p=\tau_p^\star, \\\lambda_{i, 1}=\lambda_{i, 2}=0, O_i=0, \forall i\in\mathbb{I}_{[1, t]}.
\end{aligned}
\end{equation}
Constraint \eqref{sdp_onestep:con4}, \eqref{lemma2_e2} and \eqref{sdp_onestep_noise:con2} are trivially satisfied with the defined candidate solution.
We need to prove that the constraint \eqref{sdp_onestep:con1} is feasible with this candidate solution for the future states $x_{t+1}, t\in\mathbb{N}$.

\textbf{Case I (time $t=1$): }Since the constraint \eqref{sdp_initial_noise:con2} is satisfied for $\gamma_p^\star, H_p^\star, L_p^\star, \tau_p^\star$, we have 
\begin{equation}\label{thm41}
 \begin{bmatrix}
    \!(A\!\!+\!\!BF_p^\star)^\top \!P_p^\star (\star)\!\!-\!\!P_p^\star\!\!+\!\!Q\!\!+\!\!{F_p^\star}^\top\! R F_p^\star  \!\!\!&(A\!\!+\!\!BF_p^\star)^\top P_p^\star\\ P_p^\star(A\!\!+\!\!BF_p^\star) \!\!&P_p^\star\!-\!cI
    \end{bmatrix}\!\!\!\prec  \!0.
\end{equation}
holds for any $(A, B)\in\Sigma_p$.
Then, following the same arguments as in \cite[inequalities (27)-(32)]{xie2026data}, we have $\|x_1\|_{P_p^\star}^2\leq \|x_0\|_{P_p^\star}^2\leq \gamma_p^\star$.
Thus, the problem \eqref{sdp_initial_noise} is feasible with the candidate solution \eqref{candi2} at time $t=1$.

\textbf{Case II (time $t\geq 2$): }
Since the constraints \eqref{sdp_onestep_noise:con2}, \eqref{lemma2_e2} and $\tau_p\geq 0$ are satisfied for the optimal solution at $t=1$, using \eqref{thm3_1}, we have
\begin{equation}\label{thm4_1}
 \begin{bmatrix}
    \!(A\!\!+\!\!BF_1^\star)^\top \!P_p^\star (\star)\!\!-\!\!P_1^\star\!\!+\!\!Q\!\!+\!\!{F_1^\star}^\top\! R F_1^\star  \!\!&(A\!\!+\!\!BF_1^\star)^\top P_p^\star\\ P_p^\star(A\!\!+\!\!BF_1^\star) \!\!&P_p^\star\!-\!cI
    \end{bmatrix}\!\!\!\prec  \!0.
\end{equation}
holds for any $(A, B)\in\tilde{S}_1\cap \Sigma_p$.
Pre- and post-multiplying \eqref{thm4_1} with $\begin{bmatrix}x_1^\top &\omega_1^\top\end{bmatrix}$ and its transpose, respectively, we have
\begin{equation}\label{thm4_2}
\begin{aligned}
    &[(A\!+\!BF_1^\star)x_1\!+\!\omega_1]^\top P_p^\star[(A\!+\!BF_1^\star)x_1\!+\!\omega_1]\!-\!x_1^\top P_1^\star x_1\\
    \leq &-x_1^\top (Q+F_1^{\star\top} R F_1^\star)x_1+c\omega_1^\top \omega_1 
\end{aligned}
\end{equation}
holds for any $(A, B)\in\tilde{S}_1\cap \Sigma_p$.
Since $(A_1, B_1)\in\tilde{S}_1\cap \Sigma_p$ and $x_2=(A_1\!+\!B_1F_1^\star)x_1\!+\!\omega_1$, \eqref{thm4_2} yields
\begin{equation}\label{thm4_3}
    \|x_2\|_{P_p^\star}^2-\|x_1\|_{P_1^\star}^2\leq -x_1^\top (Q+F_1^{\star\top} R F_1^\star)x_1+c\omega_1^\top \omega_1.
\end{equation}
Since $\omega_t^\top \omega_t\leq\frac{1}{\lambda_{\min}(G)}\|\omega_t\|_G^2\leq \frac{1}{\lambda_{\min}(G)}$, \eqref{thm4_3} yields 
\begin{equation}\label{thm4_31}
\begin{aligned}
    \|x_2\|_{P_p^\star}^2-\|x_1\|_{P_1^\star}^2&\leq -x_1^\top (Q+F_1^{\star\top} R F_1^\star)x_1+\frac{c}{\lambda_{\min}(G)}\\
    &\leq -\|x_1\|_Q^2+\frac{c}{\lambda_{\min}(G)}.
\end{aligned}
\end{equation}
Since $P_p^\star$ is only a candidate solution of problem \eqref{sdp_initial_noise} at time $t=1$, while $P_1^\star$ is an optimal solution, we have $\|x_1\|_{P_1^\star}^2\leq \|x_1\|_{P_p^\star}^2$.
Using \eqref{thm4_31}, $\|x_1\|_{P_1^\star}^2\leq \|x_1\|_{P_p^\star}^2$ and the upper bound on $\|x\|_{P_p^\star}^2$, we have
\begin{equation}\label{thm4_32}
\begin{aligned}
    \|x_2\|_{P_p^\star}^2-\|x_1\|_{P_p^\star}^2&\leq -\|x_1\|_Q^2+\frac{c}{\lambda_{\min}(G)}\\
    &\leq -\frac{\lambda_{\min}(Q)}{c}\|x_1\|_{P_p^\star}^2+\frac{c}{\lambda_{\min}(G)}.
    \end{aligned}
\end{equation}
Subtracting $c_{RPI}-\|x_1\|_{P_p^\star}^2$ from both sides of the inequality \eqref{thm4_32}, we have
\begin{equation}\label{theorem2:e4}
\begin{aligned}
    \|x_2\|_{P_p^\star}^2\!\!-\!c_{RPI}\leq\!(1\!-\!\frac{\lambda_{\min}(Q)}{c}\!)(\|x_{1}\|_{P_p^\star}^2\!\!-\!c_{RPI}).
\end{aligned}
\end{equation}
As $c>\lambda_{\min}(Q)$, we derive $1>1-\frac{\lambda_{\min}(Q)}{c}>0$.
If $\|x_{1}\|_{P_p^\star}^2\geq c_{RPI}$, we have 
\begin{equation}\nonumber
    \|x_2\|_{P_p^\star}^2-c_{RPI}\leq \|x_1\|_{P_p^\star}^2-c_{RPI},
\end{equation}
which further gives $\|x_2\|_{P_p^\star}^2\leq \|x_1\|_{P_p^\star}^2\leq \gamma_p^\star$.
If $\|x_{1}\|_{P_p^\star}^2\leq c_{RPI}$, we have $\|x_2\|_{P_p^\star}^2\leq c_{RPI}\leq \|x_0\|_{P_p^\star}^2\leq \gamma_p^\star$.
Thus, the problem \eqref{sdp_initial_noise} is feasible with the candidate solution \eqref{candi2} at time $t=2$.
Applying this argument inductively, we prove the problem \eqref{sdp_onestep} is feasible with the candidate solution \eqref{th3:p11} for time $t\in\mathbb{I}_{[3, \infty]}$.

\textbf{Part II:} We define $\beta=1-\frac{\lambda_{\min}(Q)}{c}$, for which $1>\beta>0$.
Since the arguments in \eqref{thm4_1}-\eqref{theorem2:e4} hold for any time $t$, we have
\begin{equation}\label{theorem2:e41}
    \|x_{t+1}\|_{P_p^\star}^2-c_{RPI}\leq \beta (\|x_{t}\|_{P_p^\star}^2-c_{RPI}),
\end{equation}
which gives that the state converges exponentially to the set $\mathcal{E}_{RPI}$.
When the state enters the RPI set, i.e., $\|x_t\|_{P_p^\star}^2\leq c_{RPI}$, since $\beta>0$ and $\|x_t\|_{P_p^\star}^2\leq c_{RPI}$, \eqref{theorem2:e41} implies that 
$\|x_{t+1}\|_{P_p^\star}\leq c_{RPI}$,
which implies that the state $x_{t+1}$ stays inside the RPI set $\mathcal{E}_{RPI}$.
%\textbf{Part III:} Since the constraints \eqref{sdp_onestep:con1}, \eqref{sdp_onestep:con4} hold for $x_t, H_t^\star, L_t^\star$. Following the same arguments as \eqref{input:nominal_lmi4}-\eqref{thm13}, we obtain that for any state $x$ satisfying
%\begin{equation}\label{thm44}
%    \begin{bmatrix}
%        x\\
%        1
%    \end{bmatrix}^\top
%    \begin{bmatrix}
%        -{\gamma_t^\star}^{-1} P_t^\star &0\\
%        0 &1
%    \end{bmatrix}
%    \begin{bmatrix}
%        x\\
%        1
%    \end{bmatrix}
%    \geq 0,
%\end{equation}
%the following inequality holds
% \begin{equation}\label{thm45}
%     \begin{bmatrix}
%         x\\
%         1
%     \end{bmatrix}^\top
%     \begin{bmatrix}
%         -(C_x+C_uF_t^\star)^\top (C_x+C_uF_t^\star) &0\\
%         0 &1
%     \end{bmatrix}
%     \begin{bmatrix}
%         x\\
%         1
%     \end{bmatrix}
%     \geq 0.
% \end{equation}
% Since the state $x_t$ satisfies \eqref{thm44}, \eqref{thm45} is satisfied for $x=x_t$, which further yields $\|(C_x+C_uF_t^\star)x_t\|=\|C_xx_t+C_uu_t\|\leq 1$.
%Suppose that $t\geq \tilde{t}$. The constraints \eqref{sdp_onestep:con4} hold for $H_p^\star$ and $L_p^\star$. 
%Since the state $x_t$ is inside the RPI set and $c_{RPI}\leq \|x_0\|_{P_p^\star}^2$, it follows that $\|x_t\|_{P_p^\star}^2\leq c_{RPI}\leq \|x_0\|_{P_p^\star}^2\leq \gamma_p^\star$.
%Using arguments analogous to those for the case $t<\tilde{t}$, we can show that $\|(C_x+C_uF_p^\star)x_t\|=\|C_xx_t+C_uu_t\|\leq 1$.
%Thus, the constraint is satisfied for the closed-loop trajectory resulting from Algorithm~2 for any $t\in\mathbb{N}$.
\end{proof}

\begin{remark}\upshape
This section proposes a data-driven min-max MPC scheme for LTV systems with process noise. 
The controller design builds on the techniques from \cite{xie2026data}, i.e., the idea of RPI set and Lyapunov function decay.
However, the technical proof requires several modifications in comparison to \cite{xie2026data}.
For example, the Lyapunov function in the proposed scheme is defined using the initial solution of the SDP problem \eqref{sdp_initial_noise} that uses the prior knowledge on the system dynamics, i.e., $P_p^\star$.
In contrast, in \cite{xie2026data} the Lyapunov function is constructed from the solution of the SDP solved at each time step.
Another difference from \cite{xie2026data} is that in the receding-horizon algorithm in \cite{xie2026data}, the SDP is no longer solved once the system state enters the RPI set. 
In the present work, however, the SDP can still be solved after the state enters the RPI set in order to further improve closed-loop performance, and closed-loop guarantees can be proved due to the different choice of the Lyapunov function.
Moreover, the data-driven characterization in Lemma~\ref{lemma2} must consider the combined effects of the variations of the system dynamics and the process noise, which is different from the LTI case in \cite{xie2026data} that only considers the effect of the process noise. 
\end{remark}

\section{Simulations}\label{sec:6}
In this section, we implement the proposed adaptive data-driven min-max MPC scheme to LTV systems without and with process noise in Sections~\ref{sec:4} and~\ref{sec:5}, respectively, using two examples.
We show that the proposed adaptive data-driven min–max MPC scheme improves closed-loop performance compared to a static state-feedback control law designed using only prior system knowledge.

\subsection{Lipschitz continuous dynamics}\label{sec:6.1}
Consider a discrete-time LTV system
\begin{equation}\label{system_simulation}
x_{t+1}=\begin{bmatrix}
    a_t &0\\
    0 &b_t
\end{bmatrix}x_t+
\begin{bmatrix}
    0.3\\ 0.1
\end{bmatrix}u_t.
\end{equation}
The parameters $a_t$ and $b_t$ are assumed to be bounded and time-varying with bounded rate of variations.
The bounds of the parameters and the rate of variations are known, i.e., 
\begin{equation}\label{bound}
\begin{aligned}
       &0.9\leq a_t\leq 1.3, \|a_{t-i}-a_{t}\|\leq 0.01i,\forall t\in\mathbb{N},\\
       &0.3\leq b_t\leq 0.5, \|b_{t-i}-b_{t}\|\leq 0.01i,\forall t\in\mathbb{N}.
\end{aligned}
\end{equation}
The system matrices $(A_t, B_t)$ satisfy the Assumption~\ref{assumption1} and the prior information is available.
The system matrices $(A_t, B_t)$ satisfy
\begin{equation}\nonumber
\begin{bmatrix}
I \!\!&A_t \!\!&B_t
\end{bmatrix}
M_{p}
\begin{bmatrix}
I \!\!&A_t \!\!&B_t
\end{bmatrix}^\top\!\succeq\! 0, \forall t\in\mathbb{N}
\end{equation}
where $M_p=\begin{bmatrix}M_{p, 11} &M_{p, 12}\\ M_{p, 12}^\top &M_{p, 22}\end{bmatrix}$ with
\begin{equation}\nonumber
\begin{aligned}
&M_{p, 11}\!\!=\!\!\begin{bmatrix}
    -1.26 \!\!&-0.03\\
    -0.03 \!\!&-0.25
\end{bmatrix},\\
&M_{p, 12}\!\!=\!\!\begin{bmatrix}
    1.1 \!\!&0 \!\!&0.3\\
    0 \!\!&0.4 \!\!&0.1
\end{bmatrix}, M_{p, 22}\!\!=\!-I.
\end{aligned}
\end{equation}
The variation of the system matrices admits the following QMI
\begin{equation}\nonumber
\begin{bmatrix}I &\Delta A &\Delta B\end{bmatrix}
M_i \begin{bmatrix}I &\Delta A &\Delta B\end{bmatrix}^\top \succeq 0
\end{equation}
with $M_i=\begin{bmatrix}M_{i, 11} &M_{i, 12}\\ M_{i, 12}^\top &M_{i, 22}\end{bmatrix}$ with
\begin{equation}\nonumber
\begin{aligned}
&M_{i, 11}=\begin{bmatrix}
    10^{-4}\times i^2 &0\\
    0 &10^{-4}\times i^2
\end{bmatrix}, \\
&M_{i, 12}=\begin{bmatrix}
    0 &0 &0\\
    0 &0 &0
\end{bmatrix}, M_{i, 22}=-I.
\end{aligned}
\end{equation}

\begin{figure}
    \centering
    \subfigure[]{\label{pic:state1}
    \includegraphics[width=0.48\textwidth]{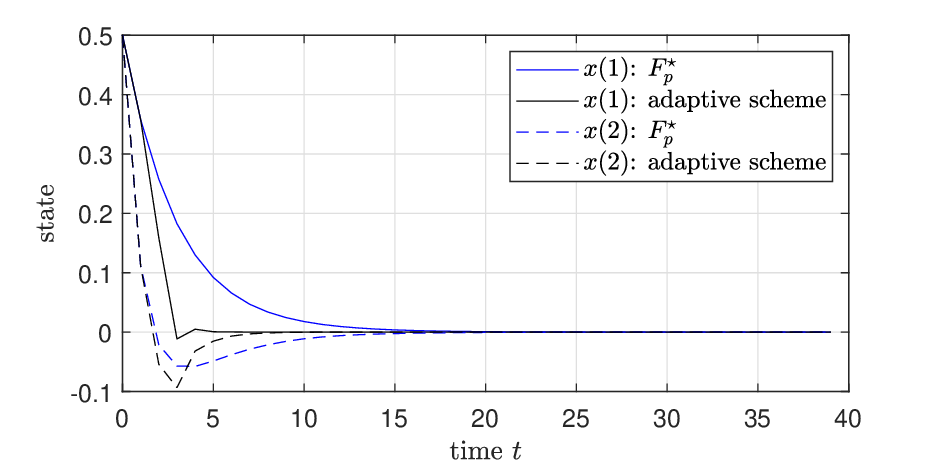}}
    \subfigure[]{\label{pic:state3}
    \includegraphics[width=0.48\textwidth]{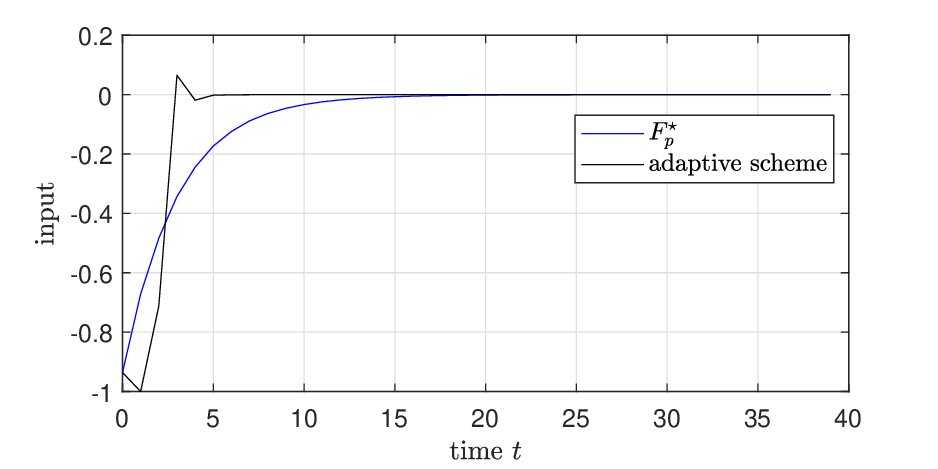}}
    \caption{Closed loop state and input trajectories under the proposed adaptive data-driven min-max MPC scheme and the static state-feedback control $F_p^\star$ for system \eqref{system_simulation} .}
    \label{pic:compare}
\end{figure}

We choose the weight matrices of the scheme as $Q=I, R=0.01I$.
The initial state is chosen as $x_0=[0.5; 0.5]$.
The matrices for constraint \eqref{constraint} are chosen as $C_x=0$ and $C_u=1$.
The system matrices $A_t$ and $B_t$ in closed-loop are drawn uniformly at random such that \eqref{bound} holds.
Figure \ref{pic:compare} illustrates the closed-loop state and input trajectories under the proposed adaptive data-driven scheme and the static state-feedback control law $F_p^\star$ over $40$ iterations.
In the adaptive data-driven scheme, the five most recent online input–state data samples are used for the data-driven characterization, as discussed in Remark~4.
The constraint is satisfied for the closed-loop system resulting from the proposed adaptive data-driven scheme and the static state-feedback control law.
While all the state and input trajectories converge to the origin, the state and input trajectories resulting from the proposed adaptive data-driven scheme converge faster than the trajectories under the static state-feedback control law $F_p^\star$.
The simulation is repeated 10 times, and the corresponding closed-loop costs are calculated.
The closed-loop cost of the proposed adaptive data-driven scheme is on average $18.55\%$ smaller than that of the static state-feedback control law $F_p^\star$.

We now consider the academic system \eqref{system_simulation} with process noise $\omega_t$.
The process noise $\omega_t$ satisfies Assumption~\ref{assumption3} with a known matrix $G=1\times 10^4 I$.
All other parameters are chosen as in the noise-free case. 
We implement the proposed data-driven min-max MPC scheme for LTV systems with process noise in Sec.~\ref{sec:5}.
Figure~\ref{pic:compare_noise} shows the closed-loop state and input trajectories under the proposed scheme and the static state-feedback control law $F_p^\star$.
The closed-loop trajectories converge to a neighborhood of the origin and satisfy the constraint.
The simulation is repeated 10 times, and the average closed-loop cost across these runs is computed.
The average closed-loop cost of the adaptive
data-driven min–max MPC scheme is $11.45\%$ smaller than that of the static state-feedback control law $F_p^\star$.

\begin{figure}
    \centering
    \subfigure[]{\label{pic:state1_noise}
    \includegraphics[width=0.48\textwidth]{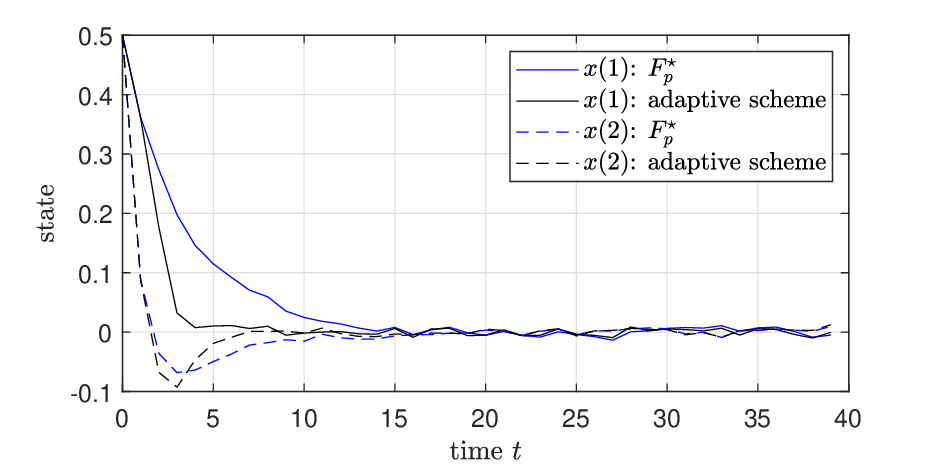}}
    \subfigure[]{\label{pic:input1_noise}
    \includegraphics[width=0.48\textwidth]{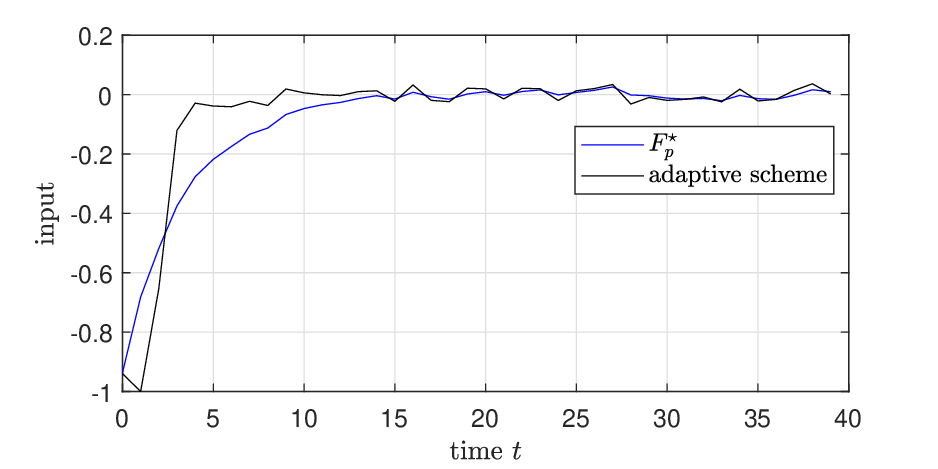}}
    \caption{Closed loop state and input trajectories under the proposed adaptive data-driven min-max MPC scheme and the static state-feedback control $F_p^\star$ for system \eqref{system_simulation} with process noise.}
    \label{pic:compare_noise}
\end{figure}

The above results show that when the SDP problems \eqref{sdp_initial} or \eqref{sdp_initial_noise} are initially feasible using the prior knowledge, the proposed adaptive data-driven min-max MPC improves closed-loop performance by using the online data.
We further consider the case with larger certainty on the time-varying parameter $b_t$, i.e., $0.2\leq b_t\leq 0.6$ for all $t\in\mathbb{N}$.
All other parameters remain the same as in the previous setting.
In this case, the SDP problem \eqref{sdp_initial} is infeasible at the initial time step, which implies that no stabilizing state-feedback controller can be designed using only the prior information.
From time $t=0$ to time $t=9$, the input is randomly selected from the interval $[-1, 1]$ to collect input-state data.
At time $t=10$, we use the prior information together with the collected online data to design an adaptive data-driven min-max MPC controller, as discussed in Remark~\ref{theorem_infeasible}.
The resulting SDP problem is feasible for all subsequent time steps.
Figure~\ref{pic:infeasible} illustrates the resulting closed-loop state and input trajectories.
Both the state and input trajectories converge to the origin and satisfy the constraint.
This result shows that the collected online data provide additional information for controller design when the prior knowledge is insufficient. 
%However, this approach relies on the feasibility of the SDP problem at all time steps.

\begin{figure}
    \centering
    \includegraphics[width=0.48\textwidth]{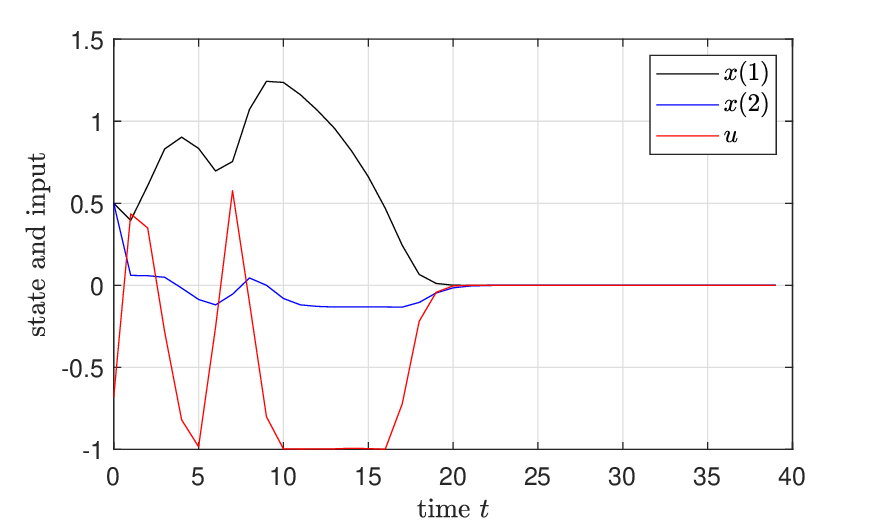}
    \caption{Closed loop state and input trajectories under the proposed adaptive data-driven min-max MPC scheme when the SDP \eqref{sdp_initial} is initially infeasible.}
    \label{pic:infeasible}
\end{figure}

\subsection{Periodic System}\label{sec:6.2}

Consider the periodic system in Example~\ref{example1}.
The matrices for the prior knowledge and bound on the variation of the system dynamics are chosen as in Example~\ref{example1}.
We choose the weighting matrix of the proposed scheme as $Q=R=I$.
The initial state is chosen as $x_0=[0.6;0.6;0.9]$.
The matrices for the constraint \eqref{constraint} is chosen as $C_x=0$ and $C_u=2$.
Figure~\ref{pic:compare_periodic} illustrates the closed-loop state and input trajectories under the proposed adaptive data-driven scheme and the static state-feedback control law $F_p^\star$ over $40$ iterations.
The closed-loop periodic system resulting from the proposed adaptive scheme converges faster to the origin compared with that of the static state-feedback control law.
The constraint is satisfied for both closed-loop trajectories.
The closed-loop cost of the proposed adaptive data-driven scheme is $17.61\%$ smaller than that of the static state-feedback control law $F_p^\star$.
This result shows that the proposed scheme is also effective for periodic LTV systems.

\begin{figure}
    \centering
    \subfigure[]{
    \includegraphics[width=0.48\textwidth]{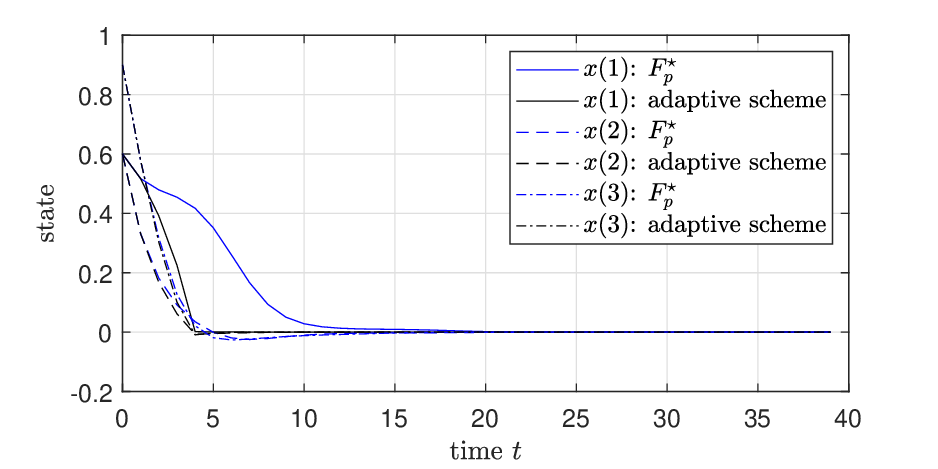}}
    \subfigure[]{
    \includegraphics[width=0.48\textwidth]{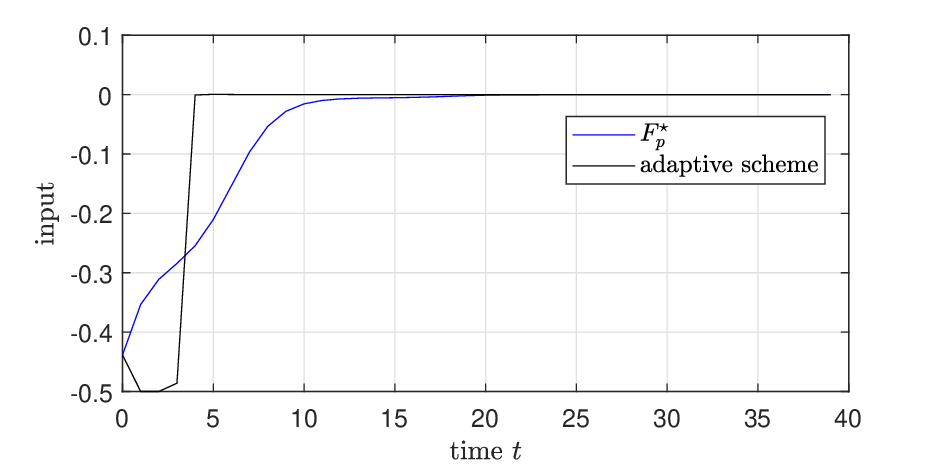}}
    \caption{Closed loop state and input trajectories under the proposed adaptive data-driven min-max MPC scheme and the static state-feedback control $F_p^\star$ for system in Example~\ref{example1}.}
    \label{pic:compare_periodic}
\end{figure}

We further consider the system in Example~\ref{example1} with process noise $\omega_t$ satisfying the Assumption~\ref{assumption3} with a known matrix $G=1\times 10^4 I$.
Figure~\ref{pic:compare_periodic_noise} illustrates the closed-loop state and input trajectories resulting from the proposed scheme in Sec.~\ref{sec:5} and the static state-feedback control law $F_p^\star$.
The constraint \eqref{constraint} is satisfied, and the state trajectories converge to a region around the origin. 
The closed-loop cost of the proposed adaptive data-driven scheme is $23.37\%$ smaller than that of the static state-feedback control law $F_p^\star$.

\begin{figure}
    \centering
    \subfigure[]{
    \includegraphics[width=0.48\textwidth]{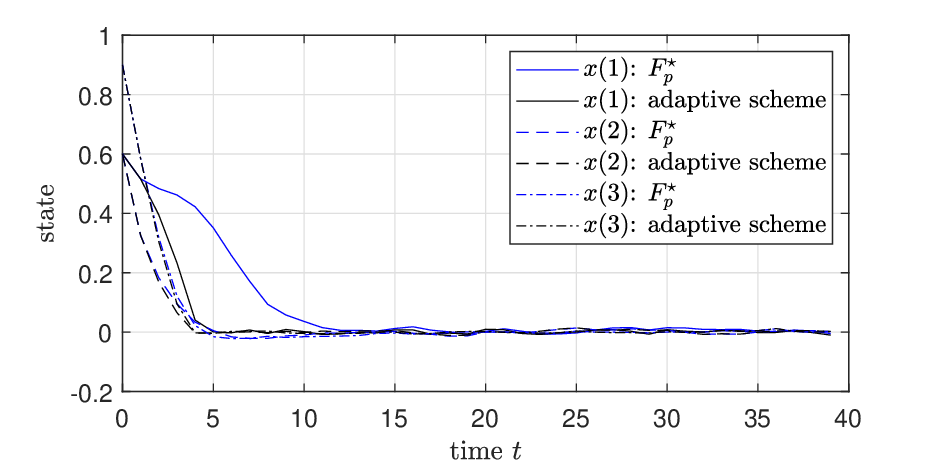}}
    \subfigure[]{
    \includegraphics[width=0.48\textwidth]{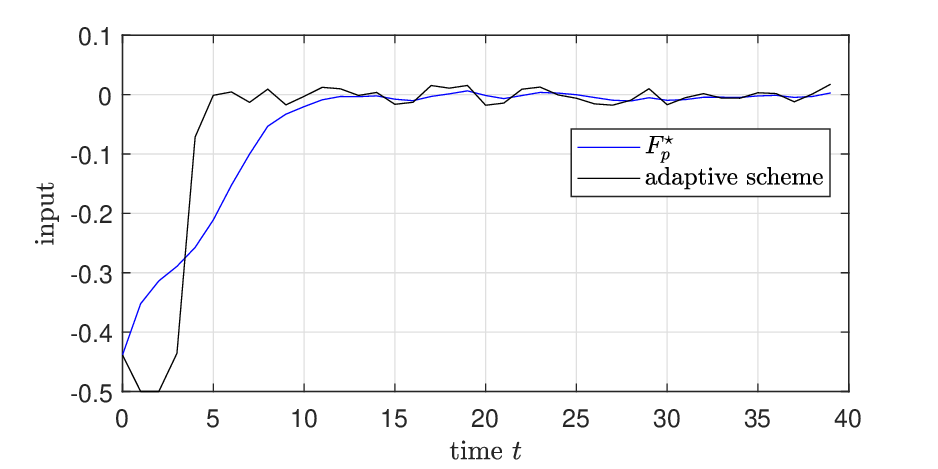}}
    \caption{Closed loop state and input trajectories under the proposed adaptive data-driven min-max MPC scheme and the static state-feedback control $F_p^\star$ for system in Example~\ref{example1} with process noise.}
    \label{pic:compare_periodic_noise}
\end{figure}

\section{Conclusion}\label{sec:7}
In this paper, we propose an adaptive data-driven min-max MPC scheme for LTV systems.
A data-driven characterization of the system dynamics at the current time step is developed using past input–state data.
The adaptive data-driven min-max MPC framework use both the prior knowledge on the system dynamics and the collected online input-state data for controller design.
For computationally tractability, the resulting control problem is reformulated as an SDP, which yields a state-feedback control law.
A receding horizon algorithm is proposed to
repeatedly solve the SDP at each time step.
We show that, when the SDP is initially feasible, the proposed scheme guarantees recursive feasibility, exponential stability and constraint satisfaction for the LTV systems.
Furthermore, we extend the framework to LTV systems subject to bounded process noise.
the data-driven system characterization and the resulting SDP account for the influence of process noise.
We establish that the proposed scheme guarantees recursive feasibility, robust stability and constraint satisfaction for LTV systems with process noise.
Simulation results show the effectiveness of the proposed method and show that incorporating online data improves closed-loop performance.

Several directions for future research remain open. 
It would be interesting to extend the proposed approach to track time-varying setpoints rather than stabilize the origin. 
Another promising direction is an development of the adaptive data-driven min-max MPC scheme for LTV systems using input-output data.

\bibliographystyle{unsrt}
\bibliography{main}

\end{document}